\RequirePackage{xcolor} 
\documentclass[journal]{IEEEtran}
\usepackage{etoolbox} 
\makeatletter
\@ifundefined{color@begingroup}%
  {\let\color@begingroup\relax
   \let\color@endgroup\relax}{}%
\def\fix@ieeecolor@hbox#1{%
  \hbox{\color@begingroup#1\color@endgroup}}
\patchcmd\@makecaption{\hbox}{\fix@ieeecolor@hbox}{}{\FAILED}
\patchcmd\@makecaption{\hbox}{\fix@ieeecolor@hbox}{}{\FAILED}

\usepackage{epsfig}
\usepackage{url}
\usepackage{amsmath}
\usepackage{amsfonts}

\usepackage{amsthm}
\usepackage{amssymb}
\usepackage{graphicx}
\usepackage{balance} 
\usepackage{algorithm2e}
\usepackage{algorithmic}
\usepackage{bm}

\newcommand{\bd}{\mathbf}
\newcommand{\C}{\mathcal{C}}
\DeclareMathOperator{\Tr}{Tr}

\newcommand{\bma}{\begin{bmatrix}}
\newcommand{\ebma}{\end{bmatrix}}
\newcommand{\Idq}{\bd{I}_{\mathrm{dq}}}
\newcommand{\Vdq}{\bd{V}_{\mathrm{dq}}}
\newcommand{\Edq}{\bd{E}_{\mathrm{dq}}}

\newcommand{\Vm}{{V}_\mathrm{dq}^2}

\newcommand{\oSone}{\overline{S}_1}
\newcommand{\oStwo}{\overline{S}_2}

\usepackage{enumitem}
\usepackage{graphicx} 
\graphicspath{{figure/}}
\usepackage{float}
\usepackage[caption=false,font=footnotesize]{subfig}
\usepackage{soul}
\usepackage{color}

\usepackage{algorithm,algorithmic}
\usepackage{cite}

\newtheoremstyle{bfnote}%
{}{}%
{\itshape}{}%
{\bfseries}{.}%
{ }%
{\thmname{#1}\thmnumber{ #2}\thmnote{ (#3)}}
\theoremstyle{bfnote}
\newtheorem{thm}{Theorem}

\newtheorem{lemma}{Lemma}

\usepackage{circuitikz}
\usepackage{tikz}
\usetikzlibrary{positioning, shapes.geometric, arrows.meta}

\title{\LARGE \bf Convexity and Optimal Online Control of Grid-Interfacing Converters with Current Limits}

\author{Lauren Streitmatter, Trager Joswig-Jones and Baosen Zhang
\thanks{Department of Electrical and Computer Engineering, University of Washington Seattle, WA 98195, USA \{lstreit, joswitra, zhangbao\}@uw.edu} %
\thanks{The authors are partially supported by the NSF grant ECCS-1942326, the Washington Clean Energy Institute, the Grainger Foundation and the Galloway Foundation. L. Streitmatter is also supported by the National Science Foundation Graduate Research Fellowship under Grant No. DGE-2140004. Any opinions, findings, and conclusions or recommendations expressed in this material are those of the authors and do not necessarily reflect the views of the National Science Foundation.}}

\begin{document}
\maketitle
\thispagestyle{empty}
\pagestyle{empty}

\begin{abstract}
Converter-based generators and loads are growing in prevalence on power grids across the globe. The rise of these resources necessitates controllers that handle the power electronic devices' strict current limits without jeopardizing stability or overly constraining behavior. Existing controllers often employ complex, cascaded control loop architecture to saturate currents, but these controllers are challenging to tune properly and can destabilize following large disturbances.

In this paper, we extend previous analysis to prove the feasible output region of a grid-connected converter is convex regardless of filter topology. We then formulate a convex optimal control problem from which we derive a projected gradient descent-based controller with convergence guarantees. This approach drives the converter toward optimality in real-time and differs from conventional  control strategies that regulate converter outputs around predefined references regardless of surrounding grid conditions. Simulation results demonstrate safe and stabilizing behavior of the proposed controller, in both the single-converter-infinite-bus systems and multi-converter networks.

\end{abstract}


\section{Introduction}
Power grids across the world are undergoing rapid transformation, driven by the increase of converter-based generation and loads. Wind turbines, solar PV, and battery storage plants that connect to the grid via power electronic devices are displacing synchronous generators to achieve decarbonization targets and reduce fuel costs. Large converter-based loads, including electric vehicle charging infrastructure and data centers, affect grid dynamics and reliability~\cite{quint2025practical}. All of these resources behave differently from conventional generation and passive loads on the grid due to their highly programmable controls and often decentralized configurations. As a result, designing converter controllers that maintain the stability and reliability of modern power networks has become an important research area~\cite{newdefinitionstability}.

One prominent control challenge of converter-based resources is their strict current limits that must be respected during both normal and contingent operations~\cite{fan2022review,baeckeland2024overcurrent}. Unlike synchronous generators which can temporarily handle more than 5 to 10 times their rated current during overloading or grid disturbances, converters can only handle currents slightly higher than their rated values before risking damage to internal semiconductor switching components~\cite{elnaggar2013comparison,ieee2022ieee,kassakian2023principles}. There is a growing body of research investigating current-limiting converter controls (see~\cite{fan2022review,baeckeland2024overcurrent,paquette2014virtual,constraintaware2026} for more).

Current-limiting capability is necessary in both grid-following (GFL) and grid-forming (GFM) controllers. GFL inverters synchronize to an existing grid voltage and regulate their output current accordingly, whereas GFM inverters establish their own voltage waveform similar to synchronous generators. Although the former is commonly used for grid-connected wind and solar generation today and more readily accommodates current-limiting strategies, it is primarily suitable for systems with low converter penetration, where a sufficiently strong grid voltage reference is available for stable synchronization~\cite{li2022revisiting}.

GFM converters address many limitations of GFL control in weak grids and systems with high converter penetration, but current limiting remains challenging. GFM control is implemented locally with limited knowledge of the surrounding network. While this is advantageous under normal operating conditions, incorporating current saturation blocks into GFM control loops causes converters to behave as current sources when saturated. This causes problems such as loss of synchronization with the grid, particularly after returning to normal unsaturated operating conditions~\cite{fan2022review, unified2026, rokrok2022}. Anti-windup strategies can help mitigate these challenges, but are sensitive to parameter tuning, reduce analytical tractability, and destabilize under large disturbances \cite{baeckeland2025}.

In this paper, we avoid the need for complicated cascaded control loops and ad-hoc current limiters. We propose an online feedback outer-loop control strategy that tracks the surrounding grid conditions to output safe and optimal converter current references even under saturation. It is compatible with GFM and GFL control paradigms and leverages the convexity of converter operating regions to guarantee convergence~\cite{streitmatter2025geometry}. 

This paper builds on the results presented in~\cite{streitmatter2025geometry} and is organized as follows. Section~\ref{sec:model} presents the network model and control problem formulation for a grid-connected converter, similar to~\cite{streitmatter2025geometry} but for the general case of any filter impedance. Section~\ref{sec:convex} proves the convexity of this converter's operating region, extending the results in~\cite{streitmatter2025geometry} to include converters with capacitor filtering. Section~\ref{sec:module} leverages convex properties to present a real-time dynamic optimal control strategy with guaranteed convergence. This extends~\cite{streitmatter2025geometry} which presents how to find optimal and safe operating points, but not how to achieve them safely in real-time. Section~\ref{sec:simulation} demonstrates the performance of the proposed controller in simulations, and Section~\ref{sec:conclusion} concludes with plans for future work.

\section{Model and Problem Formulation} \label{sec:model}
\subsection{Converter Model}
We consider a grid-connected, three-phase device that can convert direct current to alternating current (inverter) or vice versa (rectifier). While the analysis applies to converters operating in either direction, we choose the inverter as an example for this study. The grid is modeled as an infinite bus behind series resistance $R_g$ and inductance $L_g$, and the converter connects to the grid through an $RLC$ filter shown in Figure~\ref{fig:circuit}. We will prove the converter's operating region is convex even when considering the filter's shunt capacitor, generalizing the result from~\cite{streitmatter2025geometry} that only considers an $RL$ filter to any filter and line with combined equivalent impedance $Z_{eq}$ as shown in Figure~\ref{fig:equiv_circuit}.

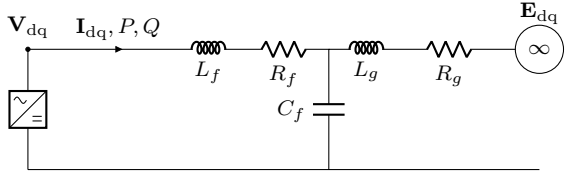
\begin{figure}[ht]
    \begin{center}
\begin{circuitikz} [american voltages,scale=0.8, font=\footnotesize]
\ctikzset{bipoles/length=7mm} 
\draw
    (0,0) to[sdcac] (0,2)
    (0,2) node[circ,label=above:$\Vdq$]{}
    -- (0.5,2)
    to[short, i=${\Idq,P,Q}$] (2.5,2)

    to[L, l_=$L_f$] (3.5,2)
    to[R, l_=$R_f$] (5,2)
    to[C, l_=$C_f$] (5,0)
    to(5,0) 

    (5,2) to[L, l_=$L_g$] (6.25,2)
    to[R, l_=$R_g$] (7.75,2)
    (8.5,2) node[circle,draw]{$\infty$}
    node[above=6pt]{$\Edq$}
    
    (7.75,2) -- (8.1,2)
    
    (0,0) -- (8.5,0);
\end{circuitikz}
\end{center}
    \caption{We model an inverter (though results hold if the inverter is replaced by a rectifier) as a controllable voltage source in the $\mathrm{dq}$ reference frame connected through an $RLC$ filter and an $RL$ line to an infinite bus.}
    \label{fig:circuit}
\end{figure}

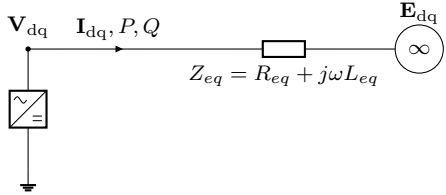
\begin{figure}[ht]
    \begin{center}
\begin{circuitikz} [american voltages,scale=0.8, font=\footnotesize]
\ctikzset{bipoles/length=7mm} 
\draw
    (0,0) node[ground]{} to[sdcac] (0,2)
    (0,2) node[circ,label=above:$\Vdq$]{}
    -- (0.5,2)
    to[short, i=${\Idq,P,Q}$] (2.5,2)

    to[generic,l_=${Z_{eq} = R_{eq} + j \omega L_{eq}}$] (6,2)

    (6.5,2) node[circle,draw]{$\infty$}
    node[above=6pt]{$\Edq$}
    
    (6,2) -- (6.1,2);
\end{circuitikz}
\end{center}
    \caption{We can instead represent the network using the equivalent impedance model of the filter and line impedance, $Z_{eq}=R_{eq} + j \omega L_{eq}$. Depending on the filter parameters, $L_{eq}$ may be positive (inductive) or negative (capacitive).}
    \label{fig:equiv_circuit}
\end{figure}

We assume the three-phase system is balanced and apply the direct-quadrature ($\mathrm{dq}$) transformation with reference to the grid voltage angle $\theta$ that rotates at fixed frequency $\omega$ to express all three-phase voltages and currents in $\mathbb{R}^2$ instead of $\mathbb{R}^3$. Further assuming the grid has constant voltage \emph{magnitude} $E$, we can express the grid voltage as $\Edq = \bma E & 0 \ebma^\top.$ Neglecting high-frequency line and filter dynamics, we express the inverter voltage as $$\Vdq=(R_{eq} + j \omega L_{eq}) \Idq + \Edq,$$ or converting the complex impedance into a real-valued matrix:

\begin{equation}\label{eq:Vdq_qss}
\Vdq = \bd{Z}_{eq} \Idq + \Edq,
\end{equation}

where $$\bd{Z}_{eq} = \bma R_{eq} & - \omega L_{eq} \\ \omega L_{eq} & R_{eq} \ebma.$$

Here, $\Vdq$ and $\Idq$ denote the inverter's voltage and current, respectively, in the grid $\mathrm{dq}$ reference frame. We use the $\mathrm{dq}$ reference frame of the grid for notational convenience, but as discussed in~\cite{streitmatter2025geometry}, convexity results and optimal setpoints do not depend on the reference angle used. Because of the linear relationship between $\Vdq$ and $\Idq$, we can consider either variable as the control or state variable. For example, updating the inverter's current output by $\Delta \Idq$ is equivalent to adjusting the voltage by $\Delta \Vdq=\bd{Z}_{eq} \Delta \Idq$ as long as there is sufficient timescale separation between inverter response times and disturbances to grid voltage $\Edq$. Thus, while we treat $\Idq$ as the system state to ease the analysis of a current-constrained controller, the same behavior can be achieved through voltage control as is common in voltage-source grids \cite{li2022revisiting}.

All control actions must keep the converter within its safe operating region to protect internal components from excessive temperatures and inductor saturation caused by overcurrent~\cite{ieee2022ieee,kassakian2023principles}. Limiting per-phase current is necessary to protect per-phase switches and inductors, but under the balanced three-phase assumption, constraining $||\Idq||$ is equivalent. Therefore, we define the safe operating region, $\mathcal{I}$, as the set of all currents with magnitude less than $I_\mathrm{max}$ and require $\Idq \in \mathcal{I}$ at all times \cite{unified2026, streitmatter2025geometry}:
$$\mathcal{I} := \left\{ \Idq \mid ||\Idq||_2^2 \leq I_\mathrm{max}^2\right\}. $$

The question of where within the safe set $\mathcal{I}$ to operate  depends on the desired outputs. Of common interest for grid-connected inverters are the device's active power output $P$, reactive power $Q$, and voltage magnitude, or equivalently, the square of the voltage magnitude $\Vm.$ Each of these quantities is a quadratic function of $\Idq$ and $\Vdq$:
\begin{subequations} \label{eq:pqv}
\begin{align}
    \label{eq:P_inverter}
    P &= \frac{3}{2} \mathbf{I}_\mathrm{dq}^\top \mathbf{V}_\mathrm{dq} \\
    \label{eq:Q_inverter}
    Q &= \frac{3}{2} \mathbf{I}_\mathrm{dq}^\top \bd J
    \mathbf{V}_\mathrm{dq} \\
    \label{eq:V2_inverter}
    \Vm &= \mathbf{V}_\mathrm{dq}^\top \mathbf{V}_\mathrm{dq},
\end{align}
\end{subequations}
where $\bd J= \bma 0 & 1 \\ -1 & 0 \ebma$. Substituting~\eqref{eq:Vdq_qss} into~\eqref{eq:pqv}, we express these quantities in terms of only the inverter current, grid voltage, and equivalent impedance:

\begin{subequations}
\begin{align*}
    P &= \frac{3}{2} \Idq^\top (\bd{Z}_{eq} \Idq + \Edq) \\
    Q &= \frac{3}{2} \Idq^\top \bd{J} (\bd{Z}_{eq} \Idq + \Edq) \\
    \Vm &= (\bd{Z}_{eq} \Idq + \Edq)^\top (\bd{Z}_{eq} \Idq + \Edq).
\end{align*}
\end{subequations}

Distributing terms and using the fact that $\bd x^T \bd M \bd x=\bd x^T \frac{\bd{M}+\bd{M}^T}{2} \bd x$ we get:
\begin{subequations} \label{eq:pqv_idq}
\begin{align}
    \label{eq:P_idq}
    P &= \frac{3}{2} R_{eq} ||\Idq||_2^2 + \frac{3}{2} \Edq^\top \Idq \\
    \label{eq:Q_idq}
    Q &= \frac{3}{2} \omega L_{eq} ||\Idq||_2^2 + \frac{3}{2} \Edq^\top \bd{J}^\top \Idq \\
    \label{eq:V2_idq}
    \Vm &= (R_{eq}^2 + \omega^2 L_{eq}^2)||\Idq||_2^2 + 2 \Edq^\top \bd{Z}_{eq} \Idq + ||\Edq||_2^2 .
\end{align}
\end{subequations}

$P$, $Q$, and $\Vm$ are all quadratic functions of the inverter's current. The next section formulates a constrained optimization problem subject to $\Idq \in \mathcal{I}$ through which we can customize a function of $P$, $Q$, and/or $\Vm$ to provide different grid services.  

\subsection{Optimal Control Problem}

Because we can control $I_\mathrm{d}$ and $I_\mathrm{q}$ (or $V_\mathrm{d}$ and $V_\mathrm{q}$) independently, we can co-optimize two of the three quantities $(S_1, S_2) \in (P, Q, \Vm).$ To track desired setpoints $\oSone$ and $\oStwo$, we can formulate the optimization problem:
\begin{equation}
    \label{eq:opt_control}
    \begin{aligned}
    \min_{\Idq} \quad & \frac12 (S_1-\oSone)^2+\gamma \frac12 (S_2-\oStwo)^2 \\
    \mbox{s.t.} \quad & \Idq \in \mathcal{I},
    \end{aligned}
\end{equation}
where $\gamma$ is a tradeoff parameter. The pair of variables we choose, $((P,Q), (P,\Vm), \ \mbox{or} \ (Q,\Vm))$, depends on which grid services we seek to provide.

\subsubsection{$P,Q$ Tracking}
Any solar-plus-storage system can implement $(P,Q)$ tracking which is desirable to control power factor or regulate local grid voltage. This is similar to how grid-following inverters operate. 

\subsubsection{$P,V^2$ Tracking}
$(P,V^2)$ tracking is similar to the operation of grid-forming inverters by simultaneously managing the device's voltage level and power output. If the active power reference is set to zero, $(P,V^2)$ tracking provides services similar to those from a static synchronous compensator (STATCOM) by absorbing or injecting reactive power as needed to regulate voltage. 

\subsubsection{$Q,V^2$ Tracking}
$(Q,V^2)$ tracking can be used to provide fault support if there is sufficient storage headroom available to increase reactive power injection and maintain local voltage levels.

The main challenge of solving \eqref{eq:opt_control} is that it is not convex in $\Idq$ due to the nested quadratics in the objective ($S_1$ and $S_2$ are functions of $||\Idq||_2^2$). Following the approach in~\cite{streitmatter2025geometry}, we rewrite the optimization problem over only $S_1$ and $S_2$ to avoid $\Idq$:
\begin{equation}
    \label{eq:opt_control_S}
    \begin{aligned}
    \min_{S_1,S_2} \quad & \frac12 (S_1-\oSone)^2+\gamma \frac12 (S_2-\oStwo)^2 \\
    \mbox{s.t.} \quad & (S_1,S_2) \in \mathcal{S}.
    \end{aligned}
\end{equation}
Here $\mathcal{S}$ is the feasible output region of the inverter given by 
\begin{equation}
    \label{eq:setS}
    \mathcal{S}:=\left\{ (S_1, S_2) \in \mathbb{R}^2 \mid \Idq \in \mathcal{I}\right\}.
\end{equation} In words, $\mathcal{S}$ is the mapping of set $\mathcal{I}$ to the $(P,Q), (P,\Vm), \ \mbox{or} \ (Q,\Vm))$ space. We prove in Section~\ref{sec:convex} that set $\mathcal{S}$ is convex regardless of filter topology or line parameters, a more general result than in~\cite{streitmatter2025geometry}. Consequently, \eqref{eq:opt_control_S} is a convex optimization problem for the single inverter-infinite bus network with an $RLC$ filter in Figure~\ref{fig:circuit} as well as for the network with generic impedance in Figure~\ref{fig:equiv_circuit}.

\subsection{Relation to grid-forming control}
In Section~\ref{sec:module}, we will introduce a converter controller that acts directly on the $I_\mathrm{d}$ and $I_\mathrm{q}$ components of the $\Idq$ vector (or equivalently $\Vdq$) through modified gradient descent of~\eqref{eq:opt_control_S}. In contrast, GFM controllers individually control voltage magnitude $||\Vdq||$ and angle $\theta_i=\arg(\Vdq)$ instead of $V_\mathrm{d}$ and $V_\mathrm{q}$. One motivation for this approach is to explicitly ensure a desirable power-frequency response around the nominal frequency $\omega_0$, similar to the behavior of synchronous generators. This can be accomplished, for example, via the droop control law~\cite{unified2026, li2022revisiting}:
\begin{equation}\label{eq:droop}
    \dot{\theta_i}=\omega_i = \omega_0 - m_P(P-\bar{P}).
\end{equation} 
The frequency controller directly increases or decreases the converter voltage angle $\theta_i$ relative to the grid voltage angle until converter power output stabilizes and $\omega_i$ synchronizes with the grid frequency.

Under our approach, controlling $V_\mathrm{d}$ ($I_\mathrm{d}$) and $V_\mathrm{q}$ ($I_\mathrm{q}$) means we do not explicitly control $\omega_i$. However, we can still compute the frequency dynamics of our controller: approximating the gradient descent-based algorithm as a continuous time controller sets $\dot{\bd V}_\mathrm{dq}$ ($\dot{\bd I}_\mathrm{dq}$) from which we can calculate $\dot{\theta_i}$. Because $\Vdq$ can be shown to converge under our approach, we know the converter frequency must still synchronize with the grid in steady-state in order for $\Vdq$ to settle at a fixed magnitude and angle. We can then interpret our implicit frequency response relative to explicit GFM frequency control laws. In doing so, we argue that optimal control laws derived from~\eqref{eq:opt_control_S} result in preferable frequency dynamics than direct frequency control (\textit{e.g.,} the linear relationship given by~\eqref{eq:droop}) as the former achieves large-signal stable, safe (no current magnitude violations), and optimal converter behavior.

\section{Convexity of the Feasible Region} \label{sec:convex}
In this section, we show the feasible output region $\mathcal{S}$ of an inverter connected to an infinite bus through equivalent impedance $Z_{eq}$ is convex, even for capacitive networks (see Figure~\ref{fig:opregion_cap}). This result is stated in Theorem~\ref{thm:main}, similar to Theorem 1 of~\cite{streitmatter2025geometry} but for the more general converter-infinite bus network model of Figure~\ref{fig:equiv_circuit}. 

\begin{thm}\label{thm:main}
    Let $(S_1,S_2)$ be a pair of points formed by choosing any two of the three quantities $P,Q,\Vm$. Let $\mathcal{S} \in \mathbb{R}^2$ be the set of all achievable points $(S_1,S_2)$ by $\Idq \in \mathcal{I}$, where $\Idq$ is the current of a converter connected to an infinite bus through some equivalent impedance $Z_{eq}$. Then $\mathcal{S}$ is convex.
\end{thm}

\begin{proof}
    To see this, we generalize the set $\mathcal{S}$ in~\eqref{eq:setS} for any two of the three quantities $P,Q,\Vm$ using the expressions in~\eqref{eq:pqv_idq} and normalize the current constraints to define an equivalent set: 
    \begin{equation} \label{eq:setC}
        \mathcal{C} := \left\{(\alpha ||\bd{x}||_2^2 + \bd{a}^\top \bd{x} + \zeta_1,\beta ||\bd{x}||_2^2 + \bd{b}^\top \bd{x} + \zeta_2) \mid ||\bd{x}||_2^2 \leq 1 \right\},
    \end{equation}
    where $\alpha,\beta,\zeta_1,$ and $\zeta_2$ are real scalars, and vectors $\bd{a}$ and $\bd{b}$ are each scalar multiples of one of $\Edq,\bd{J}\Edq,$ or $\bd{Z}_{eq}^\top \Edq$. We observe that $\bd{a}$ and $\bd{b}$ are always linearly independent: $\Edq$ and $\bd{J}\Edq$ are orthogonal, and $\bd{Z}_{eq}^\top \Edq$ is a linear combination of $\Edq$ and $\bd{J}\Edq$. The vector $\bd{x}\in \mathbb{R}^2$ represents our state vector $\Idq$. 
    
    Then recall Lemma 1 from~\cite{streitmatter2025geometry}, noting that its proof does not depend on the sign of scalars $\alpha$ and $\beta$ so can be stated for any $\alpha, \beta \in \mathbb{R}$: \textit{Given linearly independent $\bd{a}, \bd{b} \in \mathbb{R}^2$, the set $\left\{ \left( \alpha ||\bd{x}||_2^2 + \bd{a}^\top \bd{x}, \beta ||\bd{x}||_2^2 + \bd{b}^\top \bd{x} \right) \mid ||\bd{x}||_2^2 \leq 1 \right\}$ is convex.}
    
    This lemma is directly applicable to~\eqref{eq:setC} (constants $\zeta_1$ and $\zeta_2$ do not affect convexity), proving $\mathcal{C}$ and thus $\mathcal{S}$ is convex. This holds irrespective of the equivalent impedance parameter values $R_{eq}$ and $L_{eq}$.
\end{proof}

\begin{figure}[ht]
\centering
    \includegraphics[width=\linewidth]{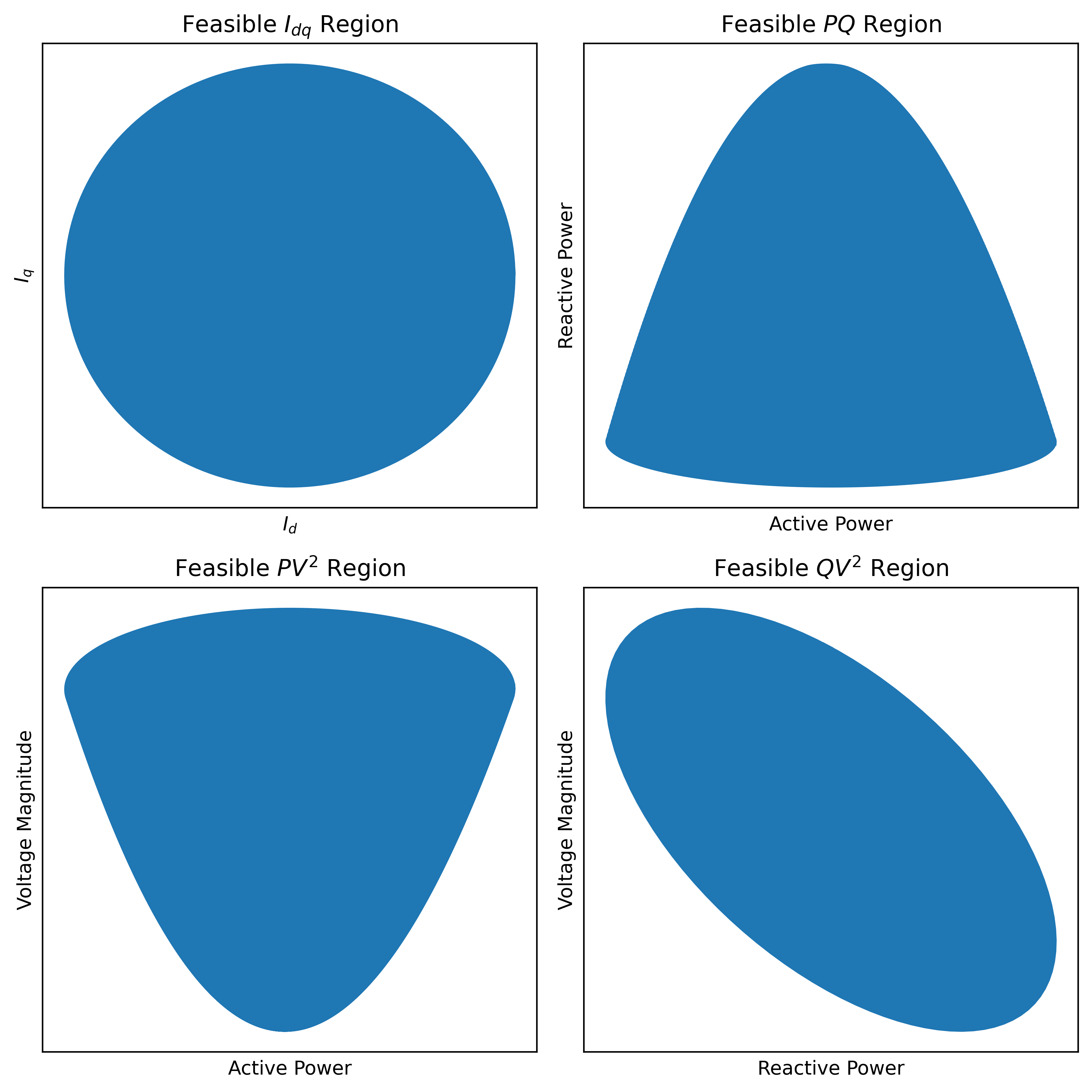}
    \caption{Limiting the current magnitude forms the safe set $\mathcal{I}$ (top left) which is a circle of radius $I_\mathrm{max}$. We transform the set of safe currents to the $\mathcal{S}=(P,Q)$ region (top right), $(P,\Vm)$ region (bottom left), and $(Q,\Vm)$ region (bottom right) for the extreme case of a converter network with capacitive reactance ($L_{eq}<0$), and we observe convexity is still preserved.}
    \label{fig:opregion_cap}
\end{figure}

\section{Dynamic Controller} \label{sec:module}
\subsection{Semi-Definite Programming Formulation} \label{sec:SDP}
Sets $\mathcal{S}$ and $\mathcal{C}$ lack direct algebraic descriptions as written, so we provide an equivalent description using linear matrix inequalities. For a 3 by 3 matrix $\bd{W}$, Lemma 2 of \cite{streitmatter2025geometry} states that the set
    \begin{equation*}
        \hat{C}=\{ (\Tr(\bd{M}_1 \bd{W}),\Tr(\bd{M}_2 \bd{W}) | \bd{W} \in \mathcal{W} \}
    \end{equation*} is equivalent to set $\C$ defined in~\eqref{eq:setC}, where
    \begin{equation} \label{eq:M12} 
    \bd{M}_1=\begin{bmatrix} \alpha \bd{I}_2 & \frac12 \bd{a} \\ \frac12 \bd{a}^T & \zeta_1 \end{bmatrix} \mbox{ and } \bd{M}_2=\begin{bmatrix} \beta \bd{I}_2 & \frac12 \bd{b} \\ \frac12 \bd{b}^T & \zeta_2 \end{bmatrix},
    \end{equation}
    \begin{equation} \label{eq:Wset} 
    \mathcal{W}=\left\{\bd{W} \in \mathbb{S}^3_+ \mid W_{11} + W_{22} \leq I_\mathrm{max}^2, W_{33} =1 \right\},
    \end{equation}
and $\bd{I}_2$ is the 2 by 2 identity matrix. Note that any rank-1 matrix $\bd W \in \mathcal{W}$ corresponds exactly to the outer product $\bma \Idq^\top & 1 \ebma^\top \bma \Idq^\top & 1 \ebma$ while rank-2 and rank-3 $\bd W$ do not have physical interpretations with respect to $\Idq$.
We then formulate~\eqref{eq:opt_control_S} as the semi-definite program (SDP): 
    \begin{subequations} \label{eq:sdp}
    \begin{align}
        \min \; & \frac12 (\Tr(\bd{M}_1 \bd{W})-\oSone)^2+\gamma \frac12 (\Tr(\bd{M}_2 \bd{W})-\oStwo)^2 \\
        \mbox{s.t. } & \bd{W} \in \mathcal{W},
    \end{align}
\end{subequations}
where $\bd M_1$ and $\bd M_2$ can be set to $\bma \frac32 R_{eq} \bd{I}_2 & \frac34 \Edq \\ \frac34 \Edq^\top & 0  \ebma$ for $P$, $ \bma \frac32 \omega L_{eq} \bd{I}_2 & \frac34 \bd J \Edq \\ \frac34 \Edq^\top \bd J^\top & 0 \ebma$ for $Q$ and $\bma (R_{eq}^2+\omega^2 L_{eq}^2) \bd{I}_2 & \bd{Z}_{eq}^\top  \Edq \\ \Edq^\top \bd{Z}_{eq} & ||\Edq||_2^2 \ebma $ for $\Vm$.

Problems~\eqref{eq:opt_control_S} and~\eqref{eq:sdp} are equivalent in that they share the same objective value. Additionally, the safe setpoints $S_1^*$ and $S_2^*$ which minimize~\eqref{eq:opt_control_S} are equal to $\Tr(\bd{M}_1 \bd{W}^*)$ and $\Tr(\bd{M}_2 \bd{W}^*),$ respectively, for $\bd{W}^*$ that minimize~\eqref{eq:sdp}. 

However, there is no guarantee that $\bd{W}^*$ is a rank-1 matrix from which we can directly extract the current $\Idq^*$ that achieves $S_1^*$ and $S_2^*$. Adding a small nuclear norm regularization term (equal to the trace for positive semi-definite matrices) to the objective function addresses this limitation and ensures rank-1 minimizers $\bd W^*$ as stated in the following lemma. The proof can be found in the appendix. 

\begin{lemma} \label{thm:rank1}
Given convex and differentiable function $f_0=f_1(\Tr(\bd M_1 \bd W))+f_2(\Tr(\bd M_2 \bd W))$, and arbitrary constant $\rho>0$, the optimal solution $\bd{W}^*$ to 
\begin{equation} \label{eq:rank1_SDP}
    \begin{aligned}
    \arg \min \quad & f_0(\Tr(\bd{M}_1\bd{W}),\Tr(\bd{M}_2\bd{W})) + \rho \Tr(\bd{W}) \\
    \mbox{s.t.} \quad & \bd W \in \mathcal{W}
    \end{aligned}
\end{equation} will have rank at most 1. 

\end{lemma}

The proof of this lemma is given in Appendix~\ref{app:proof_rank1}. The lemma provides a direct method to find $\Idq^*$ at optimality, but we need an alternative method to extract $\Idq$ from $\bd W$ when $\bd{W}$ is not optimal and possibly high rank. For example, a controller derived from \eqref{eq:rank1_SDP} may update $\bd{W}^k \rightarrow \bd{W}^{k+1}$ until convergence to $\bd{W}^*$. Hardware implementation requires physical current values $\Idq^k$ that correspond to $\bd{W}^k$ for each $k$.

\subsection{Extract $\Idq$ from high rank $\bd W$} \label{sec:findIdq}

For any $\bd W$, we know $S_1=\Tr(\bd M_1 \bd W)$ and $S_2=\Tr(\bd M_2 \bd W).$ From~\eqref{eq:pqv_idq} we have a system of two quadratic equations in $\Idq$, or $\bd x$:
\begin{equation} \label{eq:quad_sys}
\bma \alpha \\ \beta \ebma ||\bd x||_2^2 + \bma \bd a^\top \\ \bd b^\top \ebma \bd x + \bma \zeta_1 \\ \zeta_2 \ebma  = \bma S_1 \\ S_2 \ebma.
\end{equation} 

We multiply both sides by $\bd P^{-1} = \bma \bd a^\top \\ \bd b^\top \ebma ^{-1},$ rearrange terms, and solve for $\bd x$ to get
\begin{equation} \label{eq:quad_sys_x}
\bd{x} = \bd d - \mu \bd c,
\end{equation} 
where $\bd{d} = \bd P^{-1} \bma \zeta_1 - S_1 \\ \zeta_2 - S_2 \ebma,$ $\bd{c} = \bd P^{-1} \bma\alpha \\ \beta \ebma,$ and 
\begin{equation} \label{eq:mu}
    \mu = ||\bd x||_2^2,
\end{equation}
the current magnitude. We substitute~\eqref{eq:quad_sys_x} for $\bd x$ into~\eqref{eq:mu} to get a single quadratic equation with one unknown scalar ($\mu$):
\begin{equation} \label{eq:quad_mu}
||\bd c||_2^2 \mu^2 - (2 \bd d^\top \bd c + 1) \mu + ||\bd d||_2^2 =0.
\end{equation} 

The quadratic formula gives the non-negative roots $\mu_1$ and $\mu_2$ that solve~\eqref{eq:quad_mu}, and we select the smaller of the two ($\mu_1$). This ensures a safe current with minimal magnitude $\bd x_1 = \Idq = \bd d - \mu_1 \bd c$ that achieves the same $S_1$ and $S_2$ as $\bd W$.

\subsection{Modified Projected Gradient Descent of the SDP}

Solving~\eqref{eq:rank1_SDP} to optimality gives a feasible, rank-1 $\bd{W}^*$ and therefore $\Idq^*$, but it does not provide a safe trajectory to reach $\Idq^*$ from an initial $\Idq^0 \neq \Idq^*$. It is not practical to solve~\eqref{eq:rank1_SDP} to optimality at each time step because $\Idq^*$ depends on the grid voltage, $\Edq$. A converter may not be able to significantly adjust its outputs immediately following a large grid disturbance, especially with filter capacitors present. An immediate response also may not be desirable as $\Edq$ may recover within a few cycles. This incentivizes a small delay between detected changes in $\Edq$ and control response. Thus, we propose the dynamic optimal controller (OC) of Algorithm~\ref{alg:controller} which descends the gradient of~\eqref{eq:rank1_SDP} and converts each $\bd W^{k+1}$ to $\Idq^{k+1}$ for hardware implementation. The OC acts on converter current so that it can directly limit current magnitude, but updating $\Idq^k$ to $\Idq^{k+1}$ can alternatively be implemented via voltage control as is common in GFM applications. This is done using the linear relationship between inverter current and voltage: $\Vdq^{k+1}=\Vdq^{k} + \bd{Z}_{eq} (\Idq^{k+1} - \Idq^k)$ (assuming $\Edq$ is fixed between $k$ and $k+1$). 

Of note is that the gradient $\nabla_{\bd W} f(\bd{W})$ of~\eqref{eq:rank1_SDP} includes $\bd{M}_1$ and $\bd{M}_2$ which require the grid voltage. $\Edq$ is likely not known in practice but can be indirectly monitored through online feedback. At each step $k$, we know the actual inverter voltage $\Vdq^k$ and power $P^k$ that result from $\Idq^k$ as well as the values we expected from~\eqref{eq:Vdq_qss} and~\eqref{eq:P_inverter} given our $\Edq$ estimate. We then update our grid voltage estimate to resolve the difference between actual and expected voltages and power outputs. Updating the $\Edq$ estimate at every cycle produces an accurate estimate after only a few iterations post-disturbance (which is supported by simulation results in Section~\ref{sec:simulation}), assuming sufficient timescale separation between controller and $\Edq$ dynamics, a standard assumption of converter control analysis \cite{positiveseq2024,rokrok2022}.

\begin{algorithm}
\caption{Optimal Control via Modified Gradient Descent\newline\footnotesize\emph{Note:} The algorithm controls $\Idq$ to easily limit currents, but each $\Idq$ update can be achieved via $\Vdq$ control instead using $\Delta \Vdq=\bd{Z}_{eq} \Delta \Idq$.}\label{alg:controller}
\begin{algorithmic}[1]
    \REQUIRE Initial current $\Idq^0 \in \mathcal{I}$, current magnitude limit $I_\mathrm{max}$, setpoints $\oSone$ and $\oStwo$, network parameters $\Edq$ and $\bd{Z}_{eq}$, objective function $f$ with tradeoff parameter $\gamma\geq0$ and regularization constant $\rho>0$, and step-size $\alpha>0$
    \ENSURE Safe steady-state current $\Idq^*$ that solves~\eqref{eq:opt_control}
    \STATE Convert $\Idq^0 \in \mathbb{R}^2$ to its rank-1 matrix representation: $\bd{W}^{0} \gets \bma \Idq^{0 \; \top} & 1 \ebma ^\top \bma \Idq^{0 \; \top} & 1 \ebma.$
    \FOR{$k=1, 2, \dots$}
        \STATE Compute the gradient of the control objective $f$ with respect to $\bd{W}^k$, using $\bd M_1$ and $\bd M_2$ as defined in~\eqref{eq:M12}: $\nabla f(\bd{W}^k) \gets (S_1^k - \oSone) \bd M_1 + \gamma (S_2^k - \oStwo) \bd M_2 + \rho \bd I_3.$ 
        
        \STATE Take a projected gradient descent step within $\mathcal{W}$: $\bd{W}_\mathrm{PGD}^{k+1} \gets \arg \min_{\bd{W} \in \mathcal{W}} || \bd{W} - (\bd{W}^k - \alpha \nabla f(\bd{W}^k))||.$ \label{step:pgd}
        
        \STATE $\bd{W}_\mathrm{PGD}^{k+1}$ is likely not rank-1, so find the smallest-magnitude $\Idq$ that achieves the same objective as $\bd W_\mathrm{PGD}^{k+1}$ by solving for $\mu$ in~\eqref{eq:quad_mu}: $\Idq^{k+1} \gets \bd{d} - \mu_1 \bd{c}.$ \label{step:rank1}
        
        \STATE Convert $\Idq^{k+1} \in \mathbb{R}^2$ to its rank-1 matrix representation: $\bd{W}^{k+1} \gets \bma \Idq^{k+1 \; \top} & 1 \ebma ^\top \bma \Idq^{k+1 \; \top} & 1 \ebma.$ \label{step:rank1_part2}
    \ENDFOR
\end{algorithmic}
\end{algorithm}

The main result of this paper is stated in Theorem~\ref{thm:algo_convergence} which guarantees that Algorithm~\ref{alg:controller} converges to an optimal $\Idq^*$. Note that the requirement for an L-smooth objective $f$ is readily satisfied, \emph{e.g.} by quadratic functions. 

\begin{thm}\label{thm:algo_convergence}
    Let $f_0$ be a convex and differentiable L-smooth function of $\Tr(\bd M_1 \bd W)$ and $\Tr(\bd M_2 \bd W)$ over convex domain $\mathcal{W}$ defined in~\eqref{eq:Wset}. Define $f(\bd{W})=f_0(\Tr(\bd M_1 \bd W),\Tr(\bd M_2 \bd W)) + \rho \Tr(\bd{W})$. For a sufficiently small step-size $\alpha>0$,  and arbitrary regularization term $\rho>0$, the optimal controller of Algorithm~\ref{alg:controller} safely converges to a globally optimal current $\Idq^*$ that minimizes $f$. Thus, the controller in Algorithm~\ref{alg:controller} is asymptotically stable. 
\end{thm}

\subsection{Proof of Theorem \ref{thm:algo_convergence}} 
We begin by stating known convergence properties of projected gradient descent (PGD). Because $f$ is convex and the domain of $\bd W$ is closed and convex, Corollary 2.2.4 from~\cite{Nesterov_2018} applies assuming we start at a feasible $\bd W^k$. The corollary states: $f(\bd W^{k+1}_\mathrm{PGD}) \leq f(\bd W^k) - \epsilon_1$ for some $\epsilon_1 \geq 0$ and $\bd W^{k+1}_\mathrm{PGD}$ defined in Algorithm~\ref{alg:controller}. It further states $\epsilon_1=0 \iff \bd W^k$ minimizes $f$ over $\mathcal{W}$. Therefore, step~\ref{step:pgd} of Algorithm~\ref{alg:controller} converges to a minimizer $\bd{W}^*.$

Next, we will show the process of adjusting $\bd W^{k+1}_\mathrm{PGD}$ to a rank-1 matrix $\bd W^{k+1}$ also converges by showing $f(\bd W^{k+1}) \leq f(\bd W^{k+1}_\mathrm{PGD}) - \epsilon_2$ for some $\epsilon_2 \geq 0$. To see this, we will interpret steps~\ref{step:rank1} and~\ref{step:rank1_part2} of Algorithm~\ref{alg:controller} as the solution to another optimization problem:

In step~\ref{step:rank1}, we find the $\Idq^{k+1}$ value that achieves the same $S_1$ and $S_2$ values as $\bd{W}^{k+1}_\mathrm{PGD}$ with minimal magnitude. This is equivalent to solving
\begin{equation} \label{eq:rank1_adjustment}
    \begin{aligned}
    \bd{\hat W}^{k+1} = \arg & \min_{\bd W \in \mathcal{W}} \; \rho \Tr(\bd{W}) \\
    \mbox{s.t.} \; & \Tr(\bd M_1 \bd{W}) = \Tr(\bd M_1 \bd{W}^{k+1}_\mathrm{PGD}) \\
    & \Tr(\bd M_2 \bd{W}) = \Tr(\bd M_2 \bd{W}^{k+1}_\mathrm{PGD}), 
    \end{aligned}
\end{equation}
if we assume the minimizer $\bd{\hat W}^{k+1}$ is rank-1 so that we can write $\bd{\hat W}^{k+1}= \bd W^{k+1} = \bma \Idq^{k+1\;\top} & 1 \ebma^\top \bma \Idq^{k+1\;\top} & 1 \ebma$ and $\Tr(\bd{\hat W}^{k+1}) = ||\Idq^{k+1}||_2^2+1$. We confirm $\bd{\hat W}^{k+1}$ is indeed rank-1 by taking the partial Lagrangian relaxation of~\eqref{eq:rank1_adjustment} leaving $\bd W \in \mathcal{W}$ as a domain constraint. This gives $\bd{\hat W}^{k+1}=$
\begin{equation} \label{eq:dual}
    \begin{aligned}
     \arg \min_{\bd W \in \mathcal{W}} \; & \nu_1 \Tr(\bd M_1 \bd W) + \nu_2 \Tr(\bd M_2 \bd W) + \rho \Tr(\bd{W}) \\
    \; & -\nu_1 \Tr(\bd M_1 \bd{W}^{k+1}_\mathrm{PGD}) - \nu_2 \Tr(\bd M_2 \bd{W}^{k+1}_\mathrm{PGD}),\\
    \end{aligned}
\end{equation}
where $\nu_1$ and $\nu_2$ are unconstrained scalar dual variables.

 Applying Lemma~\ref{thm:rank1} again, we have $\bd{\hat W}^{k+1}$ is at most rank-1 $\implies \bd{\hat W}^{k+1}=\bd W^{k+1}$. In other words, the rank-1 adjustment of steps~\ref{step:rank1} and~\ref{step:rank1_part2} finds $\bd W^{k+1}$ to be the lowest-trace $\bd W$ that achieves the same $S_1$ and $S_2$ values as $\bd W^{k+1}_\mathrm{PGD}$ from step~\ref{step:pgd}. This ensures $f(\bd W^{k+1}) \leq f(\bd W^{k+1}_\mathrm{PGD}) - \epsilon_2$ where $$\epsilon_2=\rho \Tr(\bd W^{k+1}_\mathrm{PGD} - \bd W^{k+1}) \geq 0$$ is zero if and only if $\bd W^{k+1}_\mathrm{PGD}$ is already rank-1.

Finally, combining the two inequalities gives:
$$f(\bd W^{k+1}) \leq f(\bd W_\mathrm{PGD}^{k+1}) - \epsilon_2 \leq f(\bd W^k) - \epsilon_1 - \epsilon_2$$
$$\implies f(\bd W^{k+1}) \leq f(\bd W^k) - \epsilon_1 - \epsilon_2,$$
where the inequality is strict and the algorithm converges only when $\bd W^{k+1}=\bd W^k=\bd W^*,$ an optimal, rank-1 solution to~\eqref{eq:rank1_SDP}. Updating $\bd W^k$ to $\bd W^{k+1}$ in this way produces a trajectory of feasible $\Idq^k$ values that can be physically implemented on hardware. This trajectory is guaranteed to converge to an optimal $\Idq^*$ that achieves the same objective value as $\bd{W}^*$ with minimum current magnitude.

\section{Simulation Studies} \label{sec:simulation}
In this section, we test the optimal controller (OC) of Algorithm~\ref{alg:controller} on the single-inverter-infinite-bus network of Figure~\ref{fig:circuit} under different disturbances to confirm theoretical results and to compare performance to a current-limited droop controller. We then demonstrate the OC viability in a larger network of inverters. 

Simulations are conducted in \texttt{Python}. The \texttt{CLARABEL} solver in \texttt{CVXPY} is used to compute the projection in step~\ref{step:pgd} of Algorithm~\ref{alg:controller}~\cite{diamond2016cvxpy}. System parameters for the single-inverter-infinite-bus network are provided in Table~\ref{tab:system_params}. In per unit, $I_\mathrm{max}$ of the inverter $=1$ pu. The code used to produce the results is available.\footnote{\url{https://github.com/lpstreitmatter/Optimal-Safe-Converter-Control.git}} 

\begin{table}
    \centering
    \caption{Power System Parameters}
    \begin{tabular}{c|c|c|c}
    \hline
        Parameter & Value & Parameter & Value \\ 
         \hline
         $S_{base}$ & $1200$ VA & $V_{base}$ & $120$ V \\
         $I_{base}$ & $3.33$ A & $\omega_{base}$ & $2 \pi 60$ rad/s \\
         $R_f$ & $0.011$ pu & $R_g$ & $0.025$ pu\\
         $L_f$ & $0.016$ pu & $L_g$ & $0.021$ pu\\
         $C_f$ & $0.014$ pu & $E$ & $1$ pu\\
         \hline
    \end{tabular}
    \label{tab:system_params}
\end{table}

\subsection{Step Change in Inverter Setpoint}
In this first scenario, we simulate the inverter's response to a new, infeasible $(\overline{P},\overline{V^2})$ setpoint. We assume fixed grid voltage in this scenario, so the optimal controller maintains an accurate estimate of $\Edq$ throughout. One practical realization of this scenario on the load-side is to regulate the voltage level of grid-connected data centers that seek to track volatile spikes in power demand without disrupting the rest of the grid. New standards requiring this behavior from large loads may soon become a reality~\cite{quint2025practical}. The proposed OC approach can provide an optimal power tracking algorithm for the data center while reducing disruptive voltage fluctuations for the rest of the grid.

The simulation setup is as follows. The inverter operates at an initial current of $\Idq^0=[0.75 \; 0.3]^\top$ (pu). The simulation timestep $\Delta t=0.002$s, and at $t_0=0.05$s, the $(\overline{P},\overline{V^2})$ setpoint changes from $(0.77,1.03)$ (which corresponds to $\Idq^0$) to $(1,1)$. We run Algorithm~\ref{alg:controller} with $\rho=0.001$, $\gamma=1$, and $\alpha=1$. In per-unit, $\gamma=1$ biases the optimal $(P^*,{V^2}^*)$ toward $\overline P$. We also note that $\alpha$ must be carefully tuned depending on $\Delta t$ and network parameters. If $\alpha$ is too large, it invalidates the small step-size assumption and prevents convergence. If $\alpha$ is too small, the inverter converges too slowly to optimality. 

For our benchmark GFM control comparison, we use a simplified model of the droop controller adapted from \cite{9254645}. The droop controller updates the frequency reference according to $$\omega^{k+1} = \omega_\mathrm{nom} - m_p (P_f^k - \overline{P}),$$ and the voltage magnitude reference according to $$V_\mathrm{dq}^{k+1}=V_\mathrm{nom} - m_q (Q_f^k - \overline{Q})$$ for $PQ$ droop or $$V_\mathrm{dq}^{k+1}=\sqrt{V_\mathrm{nom}^2 - m_{V^2} (V_f^{2,k} - \overline{\Vm})}$$ for $P\Vm$ droop. Droop parameter $m_p$ is set to 0.5 Hz droop and parameters $m_q$ and $m_{V^2}$ are both set to $5\%$ voltage droop. $P_f, Q_f,$ and $V_f^2$ are filtered values of $P, Q,$ and $\Vm$ to avoid algebraic loops. At each time step, the updated voltage magnitude and frequency references produce a new current $\Idq^{k+1}$ that passes through a saturator to ensure device limits are not violated:
$$\Idq^{k+1}=\mathrm{sat}(\Idq^{k+1})=\frac{I_\mathrm{max}}{\max(I_\mathrm{max},||\Idq^{k+1}||)} \cdot \Idq^{k+1}.$$
This final, saturated current determines the droop-controlled inverter's actual voltage and power outputs.

Simulation results are shown in Figure~\ref{fig:compare_pv2_setpoint_change}. Both the optimal controller (OC) and droop controller (DR) respond safely to the change in $(\overline{P},\overline{V^2})$ setpoint, maximizing inverter current magnitude without exceeding $I_\mathrm{max}$. The optimal controller quickly converges to an optimal and safe operating point of $(0.99,1.05)$ and remains there. However, the droop controller's current limiter introduces a mismatch between the voltage and current reference signals that causes the device to oscillate around the boundary of the feasible region, failing to settle at a new operating point. If the simulation time was extended, the droop controller would continue its trajectory around the $(P,V^2)$ region.

\begin{figure}[ht]
\centering
\includegraphics[width=0.75\linewidth]{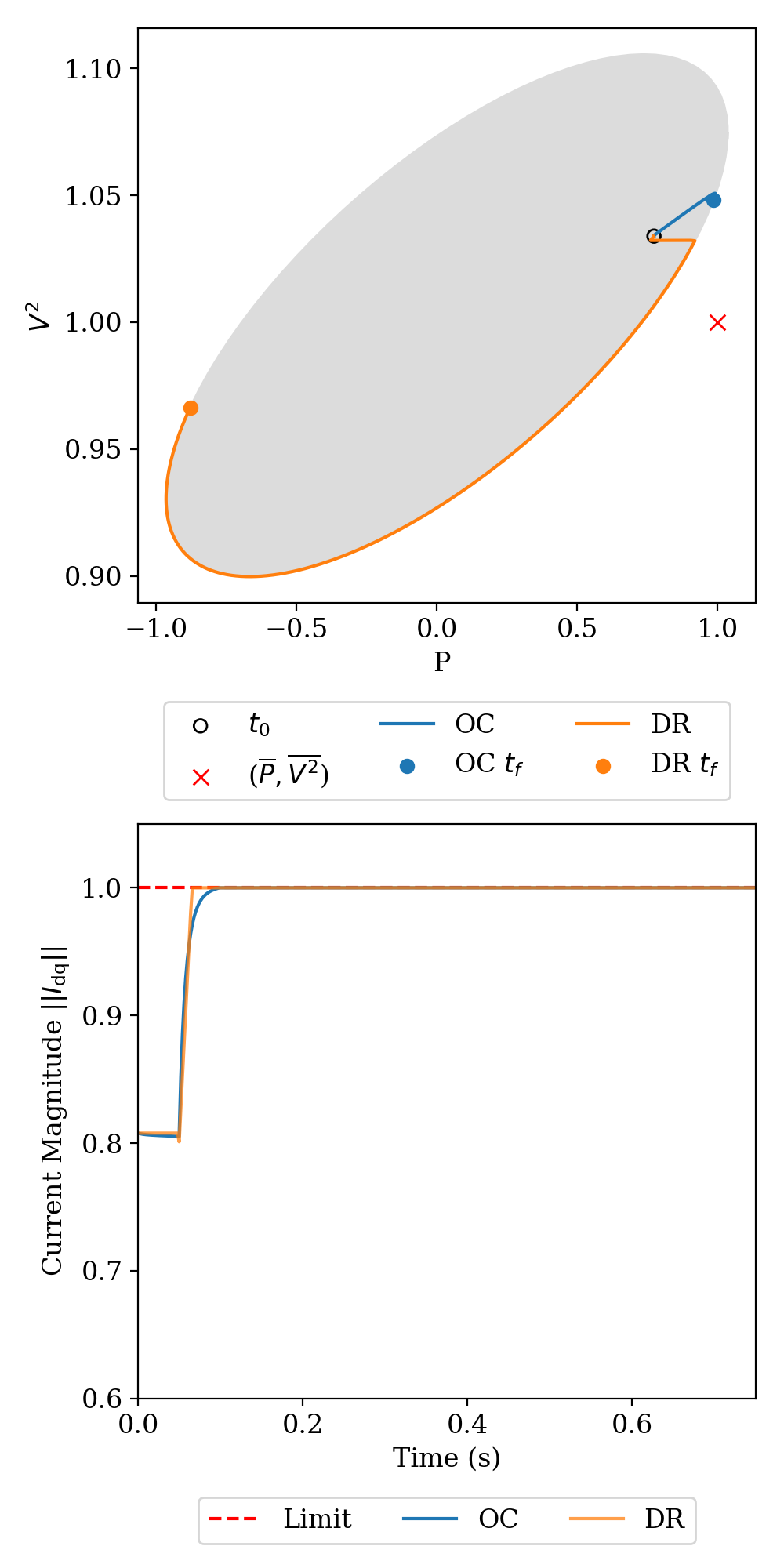}
\caption{Comparisons between the optimal controller (blue) and droop (orange) inverters' $(P,V^2)$ trajectories (top) and current magnitude responses over time (bottom) following a setpoint change at $t_0=0.05s$ are shown above. While neither controller violates current magnitude constraints, the droop controller fails to converge.}
\label{fig:compare_pv2_setpoint_change}
\end{figure}

Now we look at the inner loop dynamics. One way to interpret our setup is that the inner loop control structure achieves the algebraic relationship between $\Idq$ and $\Vdq$ in~\eqref{eq:Vdq_qss}, and our controller operates on top of it. Alternatively, our controller can be achieved via single-loops, similarly to the droop controller in~\cite{du2019comparative}. 


For the rest of the simulations, we validate our controller in multi-loop setting. We simulated the OC response in the same scenario but with inner voltage and current control loops~\cite{unifi_tutorial,rokrok2022}. Results in Figure~\ref{fig:inner_outer} show that the OC with and without inner-loop controls follow roughly the same trajectories and converge at similar times. This behavior requires sufficient timescale separation of the outer OC loop and the inner voltage and current control loops, as well as careful tuning of $\alpha$ relative to these timescales. Because we observe comparable behavior, the remaining simulations do not include inner loop controls.

\begin{figure}[ht]
\centering
\includegraphics[width=0.75\linewidth]{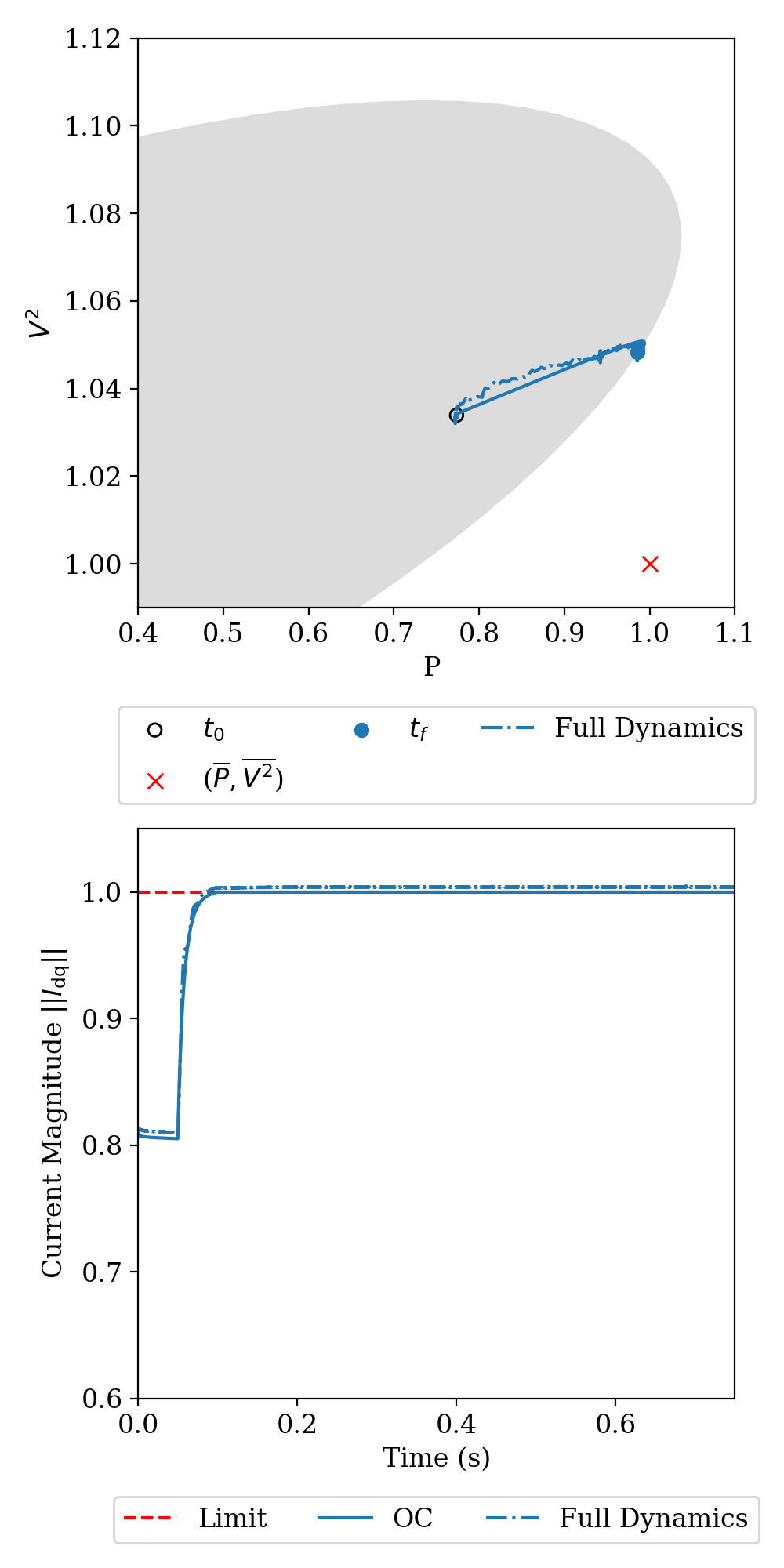}
\caption{The dashed blue line is the OC trajectory considering inner voltage and current loop dynamics compared to the OC trajectory of Figure~\ref{fig:compare_pv2_setpoint_change} shown again in solid blue for reference. Inner loop dynamics introduce noise but maintain overall trajectory and convergence behavior, validating the assumption that these dynamics can be neglected with sufficient timescale separation.}
\label{fig:inner_outer}
\end{figure}

\subsection{Step Change in Grid Voltage}
In the grid disturbance scenario, we simulate the converter's response tracking a constant $(\overline{P},\overline{V^2})$ setpoint following a 17$\%$ drop in grid-voltage magnitude $||\Edq||$ at $t_0=0.05s$. Because of the change in $\Edq$, the OC requires a few cycles for its $\Edq$ measurements to converge to the new value. We model the controller’s $\Edq$ measurement noise by adding independent zero-mean Gaussian noise to the $\mathrm{d}$ and $\mathrm{q}$ components, with identical decaying variance and zero covariance. The initial variance is $0.1 ||\Edq||$.

Results are provided in Figure~\ref{fig:compare_pv2_disturbance}. We first remark how the inverter's feasible operating region depends on the grid voltage: after $||\Edq||$ decreases, the range of feasible $\Vm$ values shifts downward and the overall $(P,V^2)$ region shrinks (compare the original grey region to the yellow region post-disturbance). As a result, the $(\overline{P},\overline{V^2})$ setpoint becomes infeasible post-disturbance; both the squared voltage magnitude and power setpoints lie outside the new feasible operating region. While both of the OC and droop controllers maintain safe current levels throughout the disturbance, the droop controller destabilizes and oscillates around the boundary of the feasible region. Meanwhile, the optimal controller converges to the closest feasible setpoint despite $\Edq$ measurement noise.

A practical example of this scenario (particularly for $(P,Q)$ tracking) is for homes with grid-connected rooftop solar. A home's installed solar capacity is often higher than its installed inverter capacity, meaning the home may request an unsafe amount of power if its load exceeds inverter ratings at times of peak solar generation. Currently, certain operators conservatively scale the apparent power setpoint according to the inverter's rated \emph{power}, not current, \emph{e.g.} $S=\sqrt{P^2 + Q^2} \leq 0.9 S_\mathrm{rated}$. This assumes the feasible operating region of the inverter is independent of the grid voltage, which we observe is not the case. The OC overcomes this limitation and maximizes the load served by rooftop generation without jeopardizing installed hardware.

\begin{figure}[ht]
    \centering
    \includegraphics[width=0.75\linewidth]{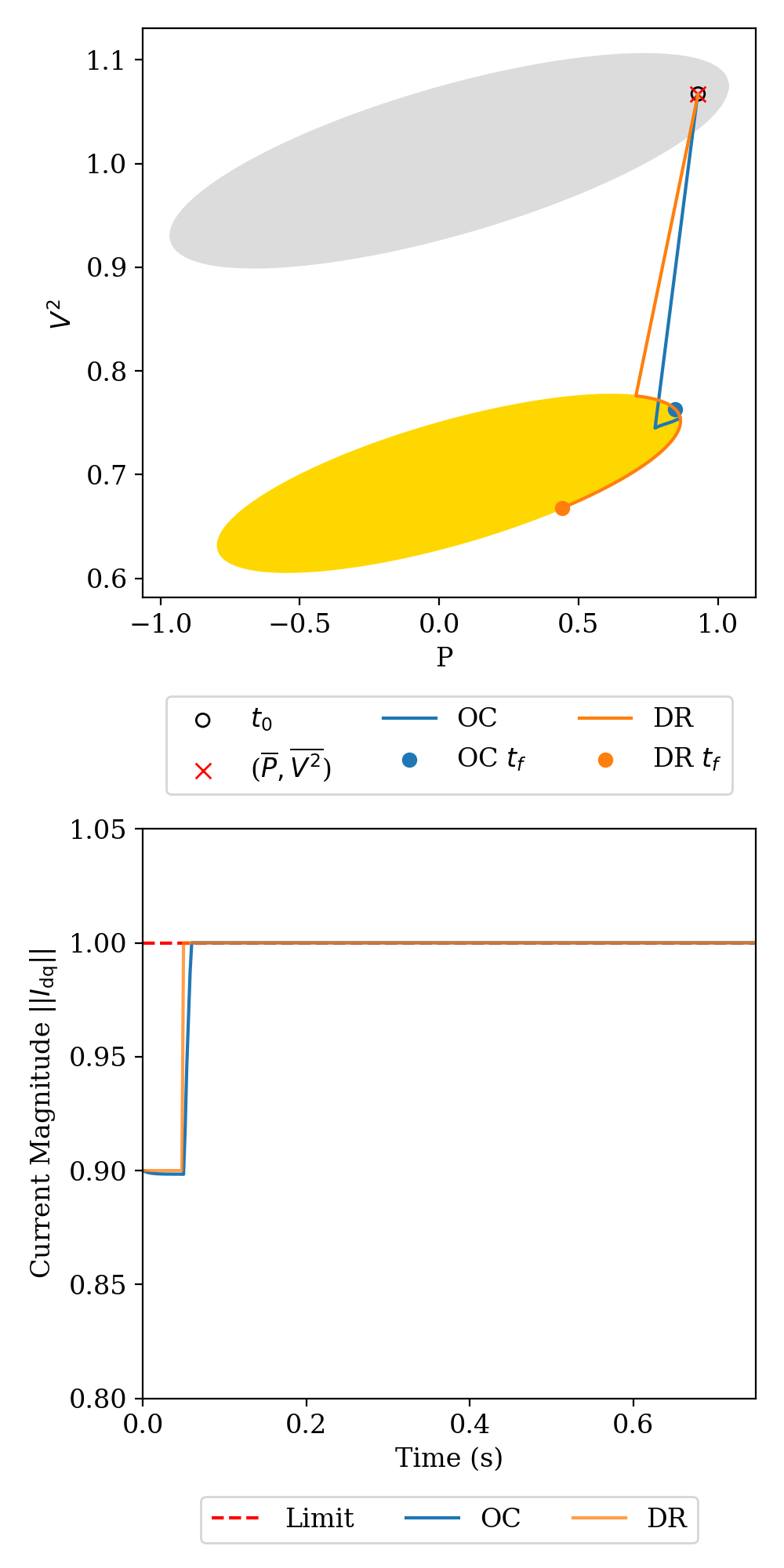}
    \caption{Comparisons between the optimal controller (blue) and droop (orange) inverters' $(P,V^2)$ trajectories (top) and current magnitude responses over time (bottom) following a step change in grid voltage at $t_0=0.05s$ are shown above. The feasible operating region $\mathcal{S}$ of the inverter \emph{before} the disturbance is in grey and \emph{after} the disturbance is in yellow. The $(\overline{P},\overline{V^2})$ setpoint becomes infeasible after the disturbance, destabilizing the saturated droop controller while the optimal controller safely settles at the nearest feasible point.}
    \label{fig:compare_pv2_disturbance}
    \end{figure}


\subsection{Multi-Inverter Network Case Study}
This final scenario demonstrates the viability of the optimal controller in a network of multiple inverters. We consider the IEEE 14-bus network of Figure~\ref{fig:ieee14_modified} modified so that buses 2, 3, 6, and 8 each have an OC inverter instead of synchronous machines~\cite{ICSEG_IEEE14Bus}. All network impedances are maintained from the original system, and each inverter connects to its bus through an $RL$ filter with $0.01 \mathrm{pu}$ resistance and $0.1 \mathrm{pu}$ reactance. We treat the slack bus (bus 1) as the grid voltage, $\Edq$ with fixed voltage magnitude and frequency. We initialize the system at a valid power flow solution then introduce a simultaneous increase in the active power setpoint of each inverter to an infeasible level $\overline{P} = 1.1 \; \mathrm{pu}$ at $t=0.05s$ while keeping $\overline{V^2}$ fixed. We use $\Delta t = 0.01s$ and set $\rho=0.001$ and $\alpha=2$. 

Each inverter's current magnitude and active power output over time is shown in Figure~\ref{fig:network_current_lowdt}. Every inverter rapidly converges to its maximum current magnitude ($1 \; \mathrm{pu}$) to produce the maximum feasible active power output given network parameters. We plot the voltage magnitude and frequency at each bus in Figure~\ref{fig:network_voltage_lowdt}. Since we do not explicitly control the frequency deviation of each inverter, we calculate it retroactively as the time derivative of $\arg(\Vdq(t))$. We observe the largest transient frequency deviation at Bus 8 of 0.5 Hz. This value is inversely proportional to $\Delta t$, so increasing $\Delta t$ to $0.05s$ reduces the peak in frequency deviation to 0.1 Hz (Figure~\ref{fig:network_voltage_highdt}). However, the tradeoff of reduced OC loop frequency, is slower inverter response times evidenced by the slower convergence of voltage magnitudes, as well as of current magnitudes and active power outputs in Figure~\ref{fig:network_current_highdt}. 

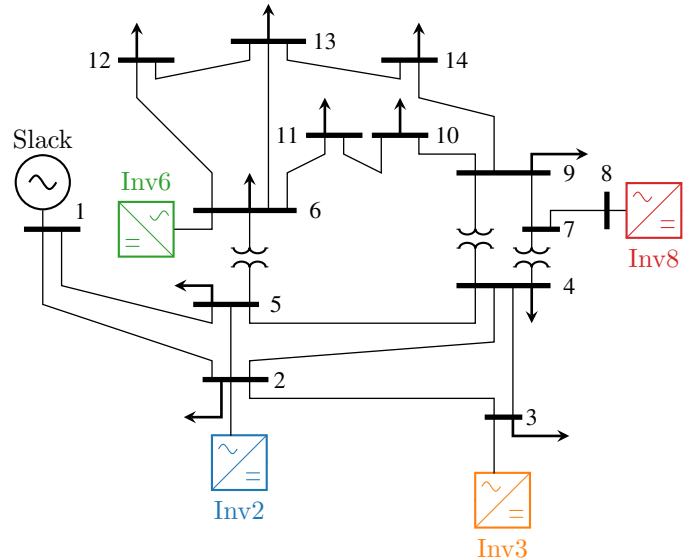
\begin{figure}[ht]
    \tikzset{flow/.style={line width=1pt, -stealth}}
\definecolor{inv2color}{RGB}{31,119,180}
\definecolor{inv3color}{RGB}{255,127,15}
\definecolor{inv6color}{RGB}{43,160,43}
\definecolor{inv8color}{RGB}{214,39,40}
\tikzset{
  smallinductor/.style={
    american inductor,
    /tikz/circuitikz/bipoles/americaninductor/coils=2,
    /tikz/circuitikz/bipoles/length=0.5cm
  }
}
\tikzset{buslabel/.style={font=\small}}
\begin{circuitikz}
    \node[buslabel] at (1.75, 4.5) {1};
	\draw[line width=2pt] (1, 4.25) -- (1.75, 4.25);
    
    \node[buslabel] at (4.4, 2.25) {2};
    \draw[line width=2pt] (3.375, 2.25) -- (4.25, 2.25);
    
    \node[buslabel] at (7.75, 1.75) {3};
    \draw[line width=2pt] (7.125, 1.75) -- (7.625, 1.75);
    
    \node[buslabel] at (8.25, 3.5) {4};
    \draw[line width=2pt] (6.75, 3.5) -- (8, 3.5);
    
    \node[buslabel] at (4.35, 3.25) {5};
    \draw[line width=2pt] (3.25, 3.25) -- (4.125, 3.25);
    
    \node[buslabel] at (4.875, 4.5) {6};
    \draw[line width=2pt] (3.25, 4.5) -| (4.625, 4.5);
    
    \node[buslabel] at (8.25, 4.25) {7};
    \draw[line width=2pt] (7.625, 4.25) -- (8.125, 4.25);
    
    \node[buslabel] at (8.75, 5) {8};
	\draw[line width=2pt] (8.75, 4.25) -- (8.75, 4.75);
    
    \node[buslabel] at (8.25, 5) {9};
    \draw[line width=2pt] (6.75, 5) -- (8, 5);
    
    \node[buslabel] at (6.625, 5.5) {10};
    \draw[line width=2pt] (5.625, 5.5) -| (6.375, 5.5);
    
    \node[buslabel] at (4.5, 5.5) {11};
    \draw[line width=2pt] (4.75, 5.5) -| (5.5, 5.5);
    
    \node[buslabel] at (2, 6.5) {12};
    \draw[line width=2pt] (2.25, 6.5) -- (3, 6.5);
    
    \node[buslabel] at (5, 6.75) {13};
    \draw[line width=2pt] (3.75, 6.75) -| (4.75, 6.75);
    
    \node[buslabel] at (6.75, 6.5) {14};
    \draw[line width=2pt] (5.75, 6.5) -| (6.5, 6.5);

	\draw[flow] (2.5, 6.5) -- (2.5, 7);
	\draw[flow] (4.25, 6.75) -- (4.25, 7.25);
	\draw[flow] (6.25, 6.5) -- (6.25, 7);
	\draw[flow] (5, 5.5) -- (5, 6);
	\draw[flow] (6, 5.5) -- (6, 6);
	\draw[flow] (7.75, 3.5) -- (7.75, 3);
	\draw[flow] (3.5, 3.25) -- (3.5, 3.5) -- (3, 3.5);
	\draw[flow] (3.625, 2.25) -- (3.625, 1.75) -- (3.125, 1.75);
	\draw[flow] (7.5, 1.75) -- (7.5, 1.5) -- (8.25, 1.5);
	\draw[flow] (7.75, 5) -- (7.75, 5.25) -- (8.5, 5.25);

	\draw (3.75, 3.75) to[smallinductor] (4.25, 3.75);
	\draw (3.75, 4) to[smallinductor,mirror] (4.25, 4);
    
	\draw[line width=0.5pt] (1.5, 4.25) -- (1.5, 3.5) -- (3.5, 3) -- (3.5, 3.25);
	\draw[line width=0.5pt] (3.75, 3.25) -- (3.75, 2.25);
	\draw[line width=0.5pt] (4, 3.25) -- (4, 3) -| (7, 3.25) -- (7, 3.5);
	\draw[line width=0.5pt] (4, 3.25) -- (4, 3.75);
	\draw[line width=0.5pt] (4, 4) -- (4, 4.5);
    
	\draw (6.75, 4) to[smallinductor] (7.25, 4);
	\draw (6.75, 4.25) to[smallinductor,mirror] (7.25, 4.25);
    
	\draw[line width=0.5pt] (7, 3.5) -- (7, 4);
	\draw[line width=0.5pt] (7, 4.25) -- (7, 5);
    
    \draw (7.5, 3.75) to[smallinductor] (8, 3.75);
	\draw (7.5, 4) to[smallinductor,mirror] (8, 4);
    
	\draw[line width=0.5pt] (7.75, 3.5) -- (7.75, 3.75);
	\draw[line width=0.5pt] (7.75, 4) -- (7.75, 4.25);
	\draw[line width=0.5pt] (7.75, 4.25) -- (7.75, 5);
	\draw[line width=0.5pt] (8, 4.25) -- (8, 4.5) -- (8.75, 4.5);
	\draw[line width=0.5pt] (4, 2.25) -- (4, 2.5) -- (7.25, 2.75) -- (7.25, 3.5);
	\draw[line width=0.5pt] (7.5, 3.5) -- (7.5, 1.75);
	\draw[line width=0.5pt] (4, 2.25) -- (4, 2) -- (7.25, 2) -- (7.25, 1.75);
	\draw[line width=0.5pt] (1.25, 4.25) -- (1.25, 3.25) -- (3.5, 2.5) -- (3.5, 2.25);
    
    \draw (1.25,5.25) to[sinusoidal voltage source, 
      /tikz/circuitikz/bipoles/length=1.2cm] 
      (1.25,4.5) node[above=20pt] {$\mathrm{Slack}$};
	\draw[line width=0.5pt] (1.25, 4.5) -- (1.25, 4.25);
    
	\draw (4.238, 1.145) to[sdcac, color=inv2color, /tikz/circuitikz/bipoles/length=1.04cm, l={\textcolor{inv2color}{$\mathrm{Inv2}$}}] (3.488, 1.145);
    
	\draw[line width=0.5pt] (3.75, 1.5) -- (3.75, 2.25);
    
	\draw (7.738, 0.645) to[sdcac, color=inv3color, /tikz/circuitikz/bipoles/length=1.04cm, l={\textcolor{inv3color}{$\mathrm{Inv3}$}}] (6.988, 0.645);
    
	\draw[line width=0.5pt] (7.25, 1) -- (7.25, 1.75);
    
	\draw (2.25, 4.25) to[sdcac, mirror, color=inv6color, /tikz/circuitikz/bipoles/length=1.04cm, l={\textcolor{inv6color}{$\mathrm{Inv6}$}}] (3, 4.25);
    
	\draw[line width=0.5pt] (3, 4.25) -- (3.5, 4.25) -- (3.5, 4.5);
    
	\draw (9.75, 4.5) to[sdcac, color=inv8color, /tikz/circuitikz/bipoles/length=1.04cm, l={\textcolor{inv8color}{$\mathrm{Inv8}$}}] (9, 4.5);
    
	\draw[line width=0.5pt] (8.75, 4.5) -- (9, 4.5);
	\draw[line width=0.5pt] (2.5, 6.5) -- (2.5, 6) -- (3.5, 5) -- (3.5, 4.5);
	\draw[line width=0.5pt] (2.75, 6.5) -- (2.75, 6.25) -- (4, 6.5) -- (4, 6.75);
	\draw[line width=0.5pt] (4.25, 6.75) -- (4.25, 4.5);
	\draw[line width=0.5pt] (4.5, 6.75) -- (4.5, 6.5) -- (6, 6.25) -- (6, 6.5);
	\draw[line width=0.5pt] (6.25, 6.5) -- (6.25, 6) -- (7.25, 5.75) -- (7.25, 5);
	\draw[line width=0.5pt] (6.25, 5.5) -- (6.25, 5.25) -- (7, 5.25) -- (7, 5);
	\draw[line width=0.5pt] (5, 5.5) -- (5, 5.25) -- (4.5, 5) -- (4.5, 4.5);
	\draw[flow] (4, 4.5) -- (4, 5);
	\draw[line width=0.5pt] (5.25, 5.5) -- (5.25, 5.25) -- (5.75, 5) -- (5.75, 5.5);  
\end{circuitikz}
    \caption{Modified IEEE 14-bus system with bus~1 as the slack bus and inverters at buses 2, 3, 6, and 8.}
    \label{fig:ieee14_modified}
\end{figure}

The fact that each converter converges to an equilibrium in the multi-converter network (regardless of convergence time which depends on control parameters) is an interesting empirical extension of Theorem~\ref{thm:algo_convergence} which guarantees convergence only for single-converter networks. We do not prove or assume the joint convexity of the feasible operating regions of multiple converters yet observe asymptotically stable OC behavior. An important next step will be to investigate the convexity of multiple-converter network configurations to determine any convergence guarantees. 

\begin{figure}[ht]
\centering
\includegraphics[width=\linewidth]{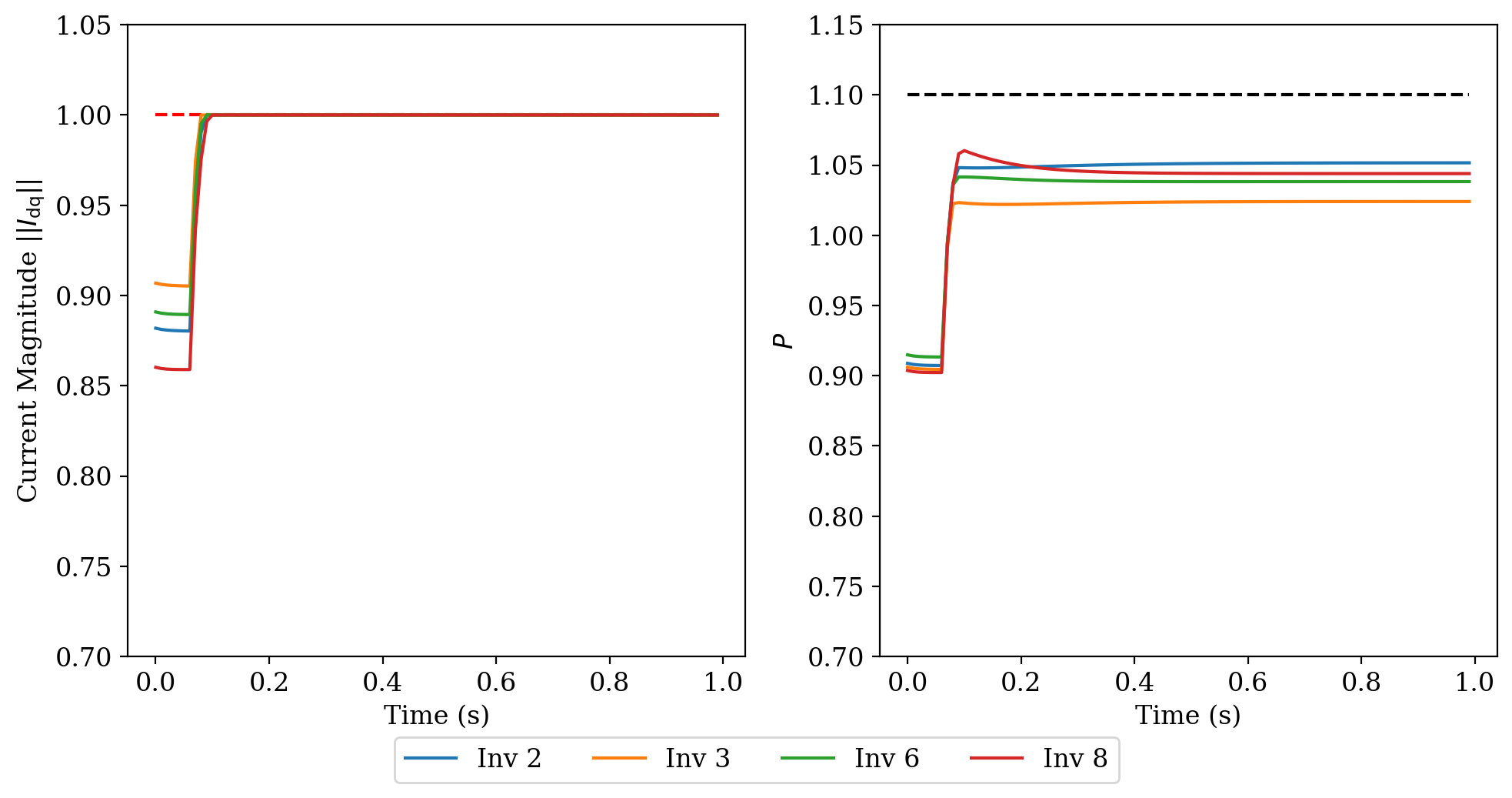}
\caption{Current magnitude (left) and active power output (right) are shown for each inverter. The dashed red line on the left plot marks the maximum safe current magnitude of each inverter ($1 \; \mathrm{pu}$) while the dashed black line on the right plot indicates the infeasible $\overline P$ introduced at $t=0.05s$. Each inverter safely converges to the maximum $1 \mathrm{pu}$ current.}
\label{fig:network_current_lowdt}
\end{figure}

\begin{figure}[ht]
\centering
\includegraphics[width=\linewidth]{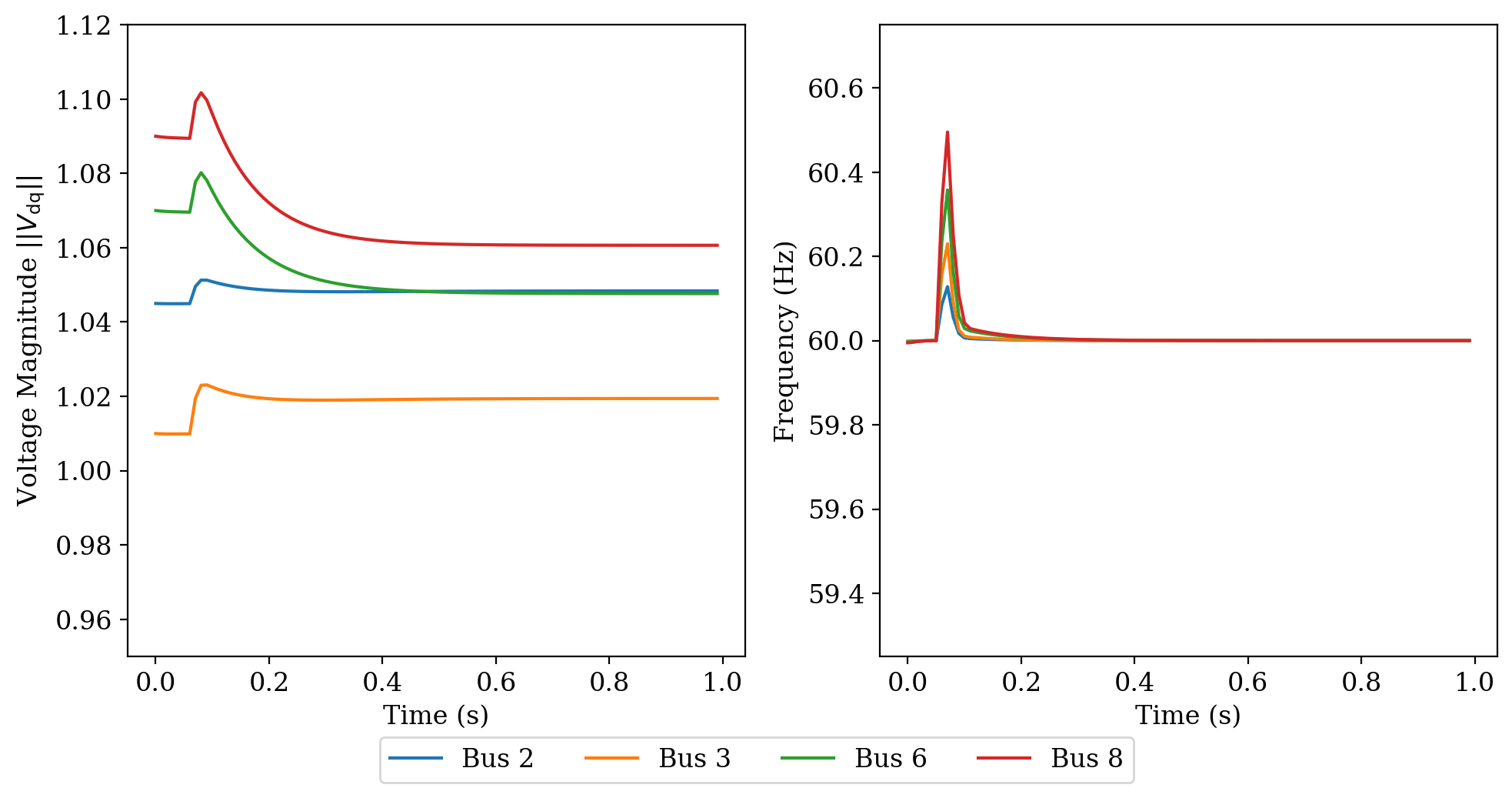}
\caption{The resulting voltage magnitude (left) and frequency (right) are shown for each of the four inverter buses. The magnitude of voltage frequency deviation is proportional to control loop frequency.}
\label{fig:network_voltage_lowdt}
\end{figure}

\begin{figure}[ht]
\centering
\includegraphics[width=\linewidth]{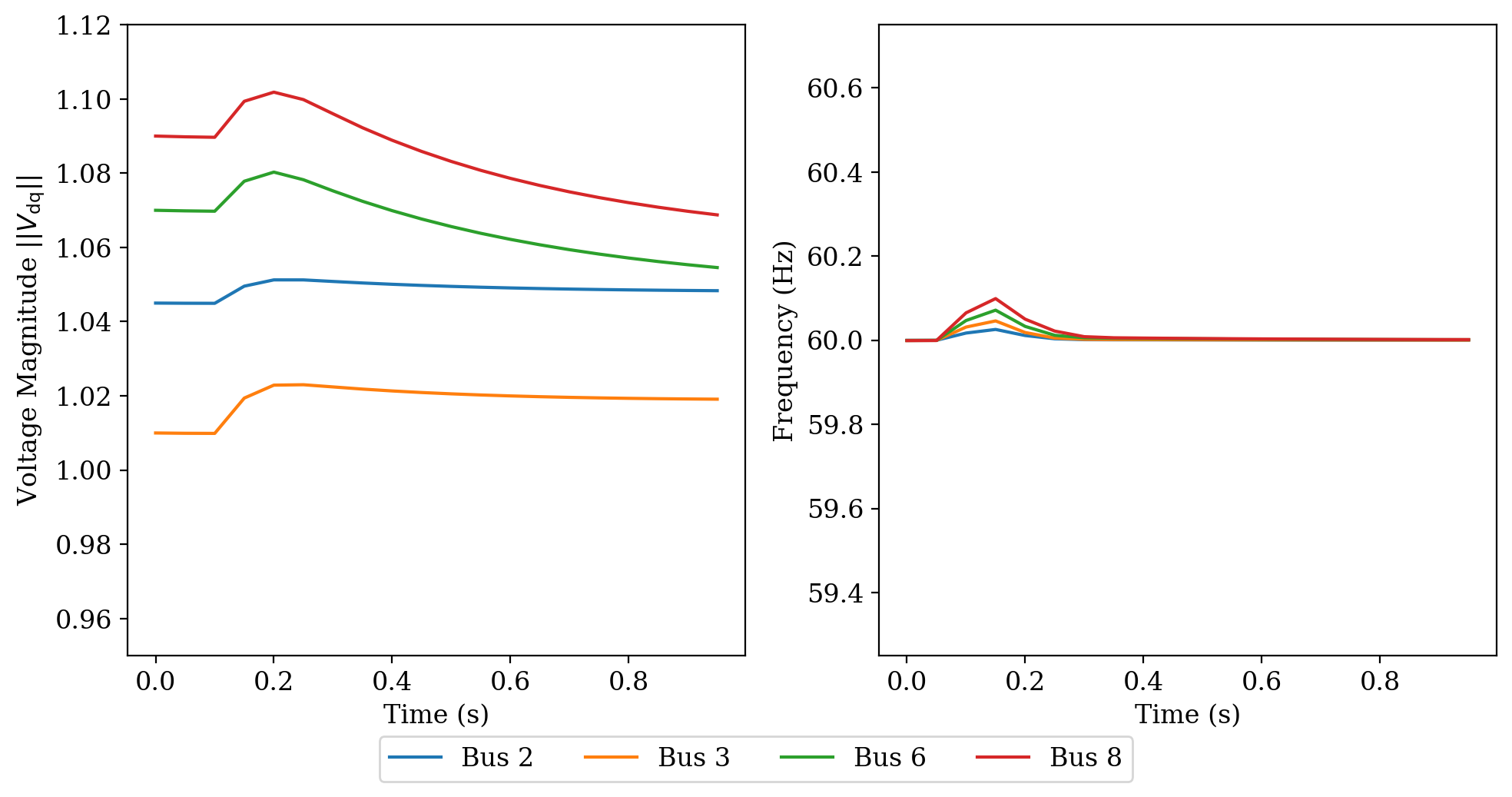}
\caption{Results are shown for the same simulation as in Figure~\ref{fig:network_voltage_lowdt} but with $\Delta t=0.05s$ instead of $0.01s$. The slower-acting control loop helps reduce the maximum frequency deviation at the cost of slower convergence times.}
\label{fig:network_voltage_highdt}
\end{figure}

\begin{figure}[ht]
\centering
\includegraphics[width=\linewidth]{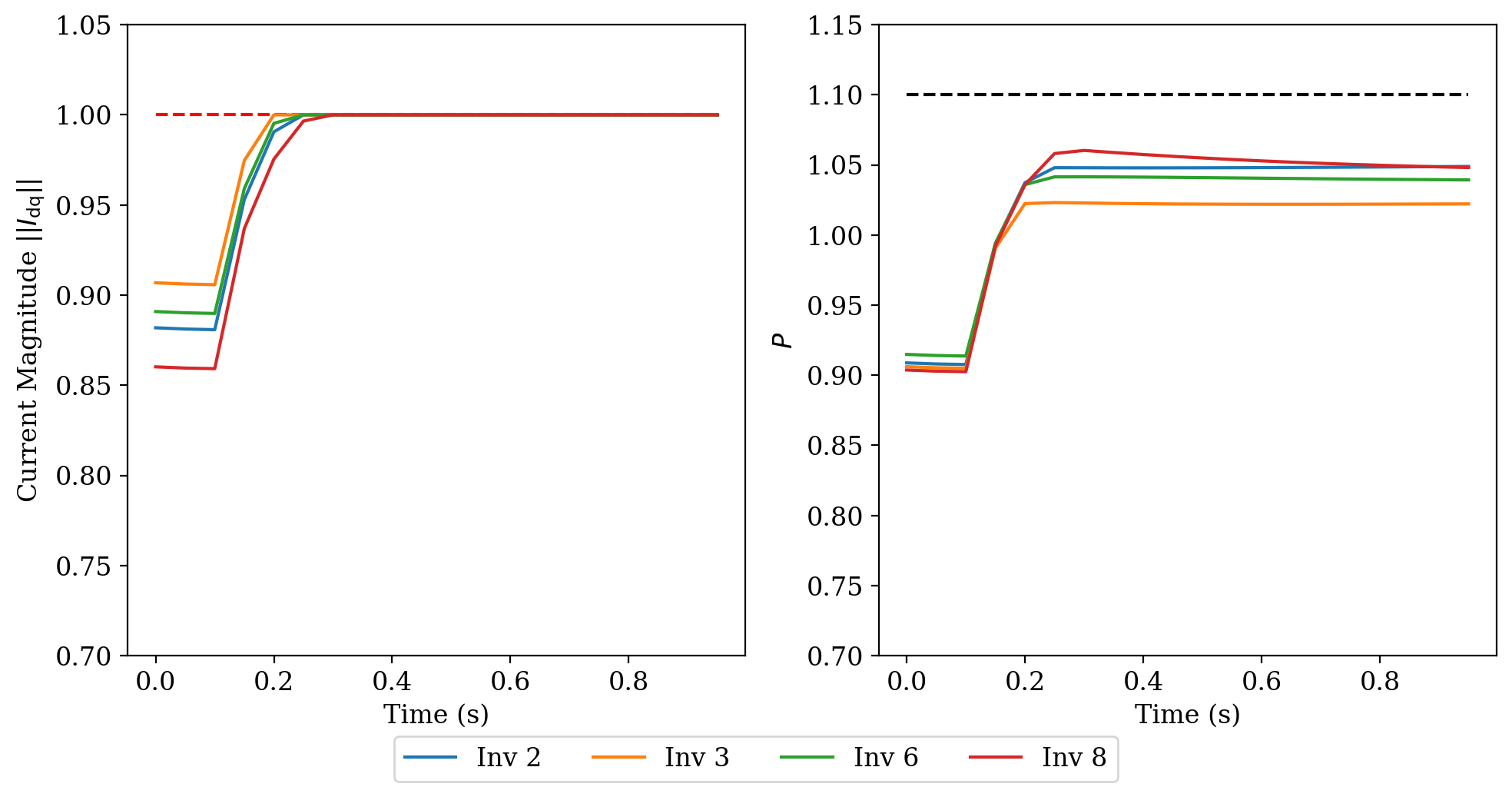}
\caption{Results are shown for the same simulation as in Figure~\ref{fig:network_current_lowdt} but with $\Delta t=0.05s$ instead of $0.01s$. Consequently, the inverter current and power outputs take longer to converge.}
\label{fig:network_current_highdt}
\end{figure}

\section{Conclusion} \label{sec:conclusion}
In this paper, we extended the proof from \cite{streitmatter2025geometry} that the feasible operating region of a grid-connected converter with an $RL$ filter is convex to grid-connected converters with generic filter topologies and parameters. While the previous work discusses how to use convexity to \emph{find} safe and optimal operating points, this paper develops a dynamic, real-time control strategy to safely \emph{achieve} optimal operating points with convergence guarantees. This approach is validated in simulations. Future work includes exploring the convexity of networked systems to design controllers with convergence guarantees even in multi-converter networks. It is also of interest to test the controller using more detailed simulation tools and on hardware to explore inner loop dynamics and bandwidth limitations.

\bibliographystyle{IEEEtran}
\bibliography{Reference}

\appendix[Proof of Lemma~\ref{thm:rank1}] \label{app:proof_rank1}
\begin{proof}
    To prove Lemma~\ref{thm:rank1}, we will take the Lagrangian of~\eqref{eq:rank1_SDP}, differentiate it with respect to $\bd{W}$, and apply the stationarity condition to find the dual variable $\bd Z^* \succeq 0$ at optimality. We will show there are three possible cases for the structure of $\bd Z^*$, and through a case-by-case analysis of each prove that the rank of $\bd Z^*$ will always be at least 2. Because complementary slackness constrains $\Tr(\bd Z^* \bd W^*)=0$ at optimality, this allows us to conclude $\bd W^*$ is always at most rank-1.

    The Lagrangian of~\eqref{eq:rank1_SDP} is given by:
    \begin{equation}\label{eq:L}
    \begin{split}
        \mathcal{L}(&\bd W, \lambda, \sigma, \bd Z) = \\
        & f_1(g_1(\bd W))+f_2(g_2(\bd W)) + \rho \Tr(\bd W) \\
        & + \lambda(W_{11} + W_{22} - I_\mathrm{max}^2) + \sigma(W_{33}-1) - \langle \bd Z, \bd W \rangle,
    \end{split}
    \end{equation}
    where $g_i (\bd W)=\Tr(\bd{M_i W})$, $\lambda \geq 0, \bd Z \succeq 0$, and $\langle \bd Z^*, \bd W^* \rangle = 0$ at optimality. The dual variables are $\lambda, \sigma,$ and $\bd Z$.

    Applying the chain rule to differentiate~\eqref{eq:L} with respect to $\bd W$ gives $\nabla_{\bd W} \mathcal{L} =$
    $$\frac{\partial g_1}{\partial \bd W}^\top \frac{\partial f_1}{\partial g_1} + \frac{\partial g_2}{\partial \bd W}^\top \frac{\partial f_2}{\partial g_2} + \bma \rho + \lambda & 0 & 0 \\ 0 & \rho + \lambda & 0 \\ 0 & 0 & \rho + \sigma \ebma - \bd Z,$$
    where $\frac{\partial f_1}{\partial g_1}, \frac{\partial f_2}{\partial g_2}$ are scalars we denote as $f_1',f_2'$ (\textit{e.g.,} the deviation from desired setpoints $\overline{S}_1, \overline{S}_2$) and $\frac{\partial g_i}{\partial \bd W} = \bd M_i$. At optimality, $\nabla_{\bd W} \mathcal{L} = 0$ and solving for $\bd Z^*$ gives:
    $$\bd Z^* = f_1' \bd M_1 + f_2' \bd M_2 + \bma \rho + \lambda & 0 & 0 \\ 0 & \rho + \lambda & 0 \\ 0 & 0 & \rho + \sigma \ebma.$$
    
    To better analyze the structure of $\bd Z^*$, we can expand $\bd M_1$ and $\bd M_2$ according to~\eqref{eq:M12}, first noting that vectors $\bd a$ and $\bd b$ are scalar multiples of two of $\left\{\Edq, \bd J \Edq, \bd Z_{eq}^\top \Edq \right\}$. Vectors $\bd a$ and $\bd b$ in~\eqref{eq:M12} can thus be written as $\bd C^\top \Edq$ and $\bd D^\top \Edq$ for matrices $\bd C, \bd D \in \left\{\bd I_2, \bd J^\top, \bd Z_{eq} \right\}$. (This simplifies analysis of Case 3 in particular.) With these terms as defined,
    \begin{equation}
    \begin{split}
    \bd Z^* = \bma (\alpha f_1' + \beta f_2') \bd I_2 & \frac12 (f_1' \bd C + f_2' \bd D)^\top \Edq \\ \frac12 \Edq^\top (f_1' \bd C + f_2' \bd D) & f_1' \zeta_1 + f_2' \zeta_2 \ebma + \\
    \bma \rho + \lambda & 0 & 0 \\ 0 & \rho + \lambda & 0 \\ 0 & 0 & \rho + \sigma \ebma.
    \end{split}
    \end{equation} 

    Recognizing that the rank of $\bd Z^*$ depends on the values of $f_1'$ and $f_2'$ (and assuming $\Edq \neq \bd 0$), we consider each of the three possible situations that may arise.

    \begin{description}
        \item[Case 1: $f_1'=0,f_2'=0$.] This occurs when the constraints on $\bd W$ are not binding (\textit{e.g.,} setpoints $\overline{S}_1$ and $\overline{S}_2$ can both be achieved exactly).
        
        In this case, $\bd Z^* = \bma \rho + \lambda & 0 & 0 \\ 0 & \rho + \lambda & 0 \\ 0 & 0 & \rho + \sigma \ebma,$ where we know $\lambda \geq 0, \rho > 0.$ Hence, $\bd Z^*$ is a diagonal matrix that will always have at least rank 2 for strictly positive $\rho$. Without $\rho$, this guarantee would not hold.

        \item[Case 2: $f_1'\neq0,f_2'=0$.] This also applies to $f_1'=0,f_2' \neq 0$ and represents the case when the unconstrained optimal value of either $f_1$ or $f_2$ is not attainable, but the other is. 
        
        In this case, $\bd Z^* = \bma (\alpha f_1' + \lambda + \rho) \bd I_2 & \frac12 f_1' \bd C^\top \Edq \\ \frac12 f_1' \Edq^\top \bd C & f_1' \zeta_1 + \sigma + \rho \ebma$ which has the form $\bma \times & 0 & x \\ 0 & \times & y \\ x & y & * \ebma$ where $\times$ denotes non-zero entries, $x,y$ are arbitrary scalars, and $*$ is an arbitrary scalar that does not impact conclusions. This divides into two sub-cases based on our assumptions: one of $x$ or $y$ is non-zero, or both $x$ and $y$ are non-zero. If the former is true, the first two columns of $\bd Z^*$ will always be orthogonal and linearly independent, ensuring the rank of $\bd Z^*$ is at least 2. 

        If the latter is true, we apply a result from graph theory to prove the rank of $\bd Z^*$ is still at least 2. Given a graph $\bd G$, we say that some matrix $\bd B$ fits $\bd G$ if for all $i \neq k$, $B_{ik}=0 \iff (i,k)$ is not an edge in $\bd G$ \cite{VANDERHOLST20031}. (Diagonal values are not constrained, hence why the value of $*$ does not impact results.) Theorem 3.4 of \cite{VANDERHOLST20031} states:

        \emph{Let $\bd G$ be a connected tree with $n$-buses. If $\bd B \succeq 0$ and $\bd B$ fits $\bd G$, then $\mathrm{rank}(\bd B) \geq n-1.$}

        We apply this to our case where we can say $\bd Z^* \succeq 0$ ``fits'' a connected 3-bus tree $\bd G$ and thus has rank $\geq 3-1=2$.
        
        \item[Case 3: $f_1' \neq 0,f_2'\neq0$.] This occurs when the constraints on $\bd W$ are binding with respect to both $f_1$ and $f_2$.
        
        In this case, it suffices to show $(f_1' \bd C + f_2' \bd D)^\top \Edq \neq \bd 0.$ Otherwise, we could have $\alpha f_1' + \beta f_2' + \lambda + \rho = 0 \implies \mathrm{rank}(\bd Z^*)=1.$ Note that a linear combination of any two matrices $\bd C, \bd D \in \left\{\bd I_2, \bd J^\top, \bd Z_{eq} \right\}$ can take on three possible forms:
        $$\left\{ \bma i & j \\ -j & i \ebma, \bma i - z_1 & z_2 \\ -z_2 & i-z_1 \ebma, \bma -z_1 & j + z_2 \\ -j-z_2 & -z_1 \ebma \right\},$$ 
        for non-zero $i,j,z_1,z_2$ representing the distinct, non-zero entries of $\bd I_2, \bd J^\top$, and $\bd Z_{eq}$. Simple calculations confirm each $2 \times 2$ matrix determinant is strictly positive, meaning $(f_1' \bd C + f_2' \bd D)$ is full rank $\implies (f_1' \bd C + f_2' \bd D)^\top \Edq \neq \bd 0.$ Therefore, we conclude that $\bd Z^*$ has at least rank 2.
    \end{description}

    Through case-by-case analysis, we have shown $\bd Z^*$ always has at least rank 2 at optimality, and therefore, by complementary slackness, an optimal $\bd W^*$ has at most rank 1.
\end{proof}

\end{document}